\documentclass[a4paper,12pt]{elsart1p}
\usepackage{amsfonts}
\usepackage[latin1]{inputenc}
\usepackage{times}
\usepackage[T1]{fontenc}
\usepackage[dvips]{graphicx}
\usepackage[dvips]{color}
\usepackage{xspace}
 \usepackage[dvips,
 colorlinks,linkcolor=webgreen,filecolor=webbrown,citecolor=webgreen,
 breaklinks=true, urlcolor=webbrown
 ]{hyperref}
\usepackage{amsmath}
\usepackage{amssymb}%
\usepackage{array}
\newtheorem{theorem}{Theorem}

\newtheorem{property}[theorem]{Property}
\newtheorem{corollary}[theorem]{Corollary}

\newtheorem{proposition}[theorem]{Proposition}
\newtheorem{remark}[theorem]{Remark}
\newenvironment{proof}[1][Proof]{\noindent\textbf{#1.} }{\ \rule{0.5em}{0.5em}}
\definecolor{webgreen}{rgb}{0,.5,0}
\definecolor{webbrown}{rgb}{.6,0,0}
\definecolor{webyellow}{rgb}{0.98,0.92,0.73}
\definecolor{webgray}{rgb}{.753,.753,.753}
\definecolor{webblue}{rgb}{0,0,.8}
\graphicspath{{./graphics/}{.}{./figs/}}
\def\dx{{dx}}
\journal{Information Sciences}
\begin{document}

%%% -->
\begin{frontmatter}
\setlength{\hsize}{16.5cm}



\title{On Some Entropy Functionals derived from Rényi Information Divergence}


%\author{J.-F. Bercher\corauthref{cor1}\thanksref{label2}}
\author{J.-F. Bercher\thanksref{label2}}
\thanks[label2]{On sabbatical leave from ESIEE-Paris, France}
\ead{jf.bercher@esiee.fr}
\address{Laboratoire des Signaux et Syst\`emes, CNRS-Univ Paris Sud-Supelec, 
91192 Gif-sur-Yvette cedex, France}
%\date{Typesetted \today - final (accepted) version}

\begin{abstract}
% Text of abstract
% In this paper/work, we consider the ME problem associated with Rényi divergences that is, we look for the distribution that maximizes the entropy subject to some moment constraint, that represent the observable (mean energy for instance) in a physical setting. 
% We study the expression of the entropy as a function of the observable/constraint. So doing, we define two entropy functionals in the object space. Of course these entropy recover in the limit, the classical rate functions of Large deviations Theory. The behaviour of these functions is documented.
We consider the maximum entropy problems associated with Rényi $Q$-entropy, subject to two kinds of constraints on expected values. The constraints considered are a constraint on the standard expectation, and a constraint on the generalized expectation as encountered in nonextensive statistics. The optimum maximum entropy probability distributions, which can exhibit a power-law behaviour, are derived and characterized. 

The Rényi entropy of the optimum distributions can be viewed as a function of the constraint. This defines two families of entropy functionals in the space of possible expected values. General properties of these functionals, including nonnegativity, minimum, convexity, are documented. Their relationships as well as numerical aspects are also discussed. Finally, we work out some specific cases for the reference measure $Q(x)$ and recover in a limit case some well-known entropies.

\end{abstract}

\end{frontmatter}
%%% -->

%\title{On Some Entropy Functionals derived from Rényi Information Divergence}

%\author{\href{mailto:jf.bercher@esiee.fr}{J.-F. Bercher}\footnote{On sabbatical leave from ESIEE-Paris, France, \href{mailto:jf.bercher@esiee.fr}{jf.bercher@esiee.fr}.}\\
%Laboratoire des Signaux et Syst\`emes, \\CNRS-Univ Paris Sud-Supelec
%\\91192 Gif-sur-Yvette cedex, France
%\\ email: \href{mailto:bercher@lss.supelec.fr}{bercher@lss.supelec.fr} }
%\date{Typesetted \today}
\maketitle

% \markboth{J.-F. Bercher}{Entropy Functionals derived from Rényi Information Divergence } \pagestyle{myheadings}

% % \section*{Abstract}
% % 
% % In this paper/work, we consider the ME problem associated with rényi divergences that is, we look for the distribution that maximizes the entropy subject to some moment constraint, that represent the observable (mean energy for instance) in a physical setting. 
% % We study the expression of the entropy as a function of the observable/constraint. So doing, we define two entropy functionals in the object space. Of course these entropy recover in the limit, the classical rate functions of Large deviations Theory. The behaviour of these functions is documented.
 
% We focus on
% % \bigskip
% % 
% % PACS: 05.40.-a; 02.50.-r;
% % % Fluctuation phenomena, random processes, noise, and Brownian motion
% % % Probability theory, stochastic processes, and statistics
% % %05.30.-d; % Quantum statistical mechanics
% % 05.20.-y; % Classical statistical mechanics
% % %05.70.Ce; %Thermodynamics
% % 05.90.+m %Other topics in statistical physics, thermodynamics, and nonlinear dynamical systems

\bigskip
% 
% Keywords: Rényi entropy

%% !! Changing line spacing
%\renewcommand{\baselinestretch}{2}\small\normalsize

%\section*{}

\setlength{\parindent}{0cm}
%%
% Graphicspath
\graphicspath{{./graphics/}{.}{./figs/}}
%%

%%%%%%%%%%%%%%%%%%%%% INTRODUCTION %%%%%%%%%%%%%%%%%%%%%%%%%%%%%%%%%%%

\section{Introduction}

Consider two univariate continuous probability distributions with densities $P$ and $Q$ with respect to the Lebesgue measure. The Rényi information divergence introduced in \cite{Renyi1961} has the form
\begin{equation}
D_{\alpha}(P||Q)=-H^{(\alpha)}_Q(P)=\frac{1}{\alpha-1}\log\int_\mathcal{D} P(x)^{\alpha}Q(x)^{1-\alpha}dx,
\label{eq:defRenyi}
\end{equation}
where $\alpha$ is a positive real and $\mathcal{D}$ the domain of definition of the integral. In the discrete case, the continuous sum is replaced by a discrete one which extends on a subset $\mathcal{D}$ of integers.   The opposite $H^{(\alpha)}_Q(P)$ of the Rényi information divergence can be viewed as a Rényi entropy relative to the reference measure $Q$, and can be called $Q$-entropy. By L'Hospital's rule, Kullback divergence is recovered in the limit $\alpha\rightarrow 1$.

Applications and areas of interest in Rényi entropy are plentiful: communication and coding theory \cite{Csiszar1995}, data mining, detection, segmentation, classification \cite{Neemuchwala2005,Basseville1989}, hypothesis testing \cite{Molina2007}, characterization of signals and sequences \cite{Vinga2004,Krishnamachari2004}, signal processing \cite{Basseville1989,Baraniuk2001}, image matching and registration \cite{Neemuchwala2005,He2003}. 
Connection with the log-likelihood has been outined in \cite{Song2001}, where is also defined a measure of the intrinsic shape of a distribution which can serve as a measure of tail heaviness \cite{Nanda2007}. Rényi entropies for large families of univariate and bivariate distributions are given in  \cite{Nadarajah2003,Nadarajah2005}. Divergence measures based on entropy functions can be used in the process of inference \cite{Esteban1995}, in clustering or partionning problems \cite{Mayoral1998,Banerjee2005,Bhandari1993}.

Rényi entropy also plays a central role in the theory of multifractals, see the review \cite{Jizba2004} and \cite{Bashkirov2004}. In statistical physics, following Tsallis proposal \cite{Tsallis1988,Tsallis2002} of another entropy (which is simply related to Rényi entropy), there has been a high interest on these alternative entropies and the development of a community in ``nonextensive thermostatistics''. Indeed,  the associated maximum entropy distributions exhibit a power-law behaviour, with a remarkable agreement with experimental data, see for instance
\cite{Beck2004,Tsallis2002} and references therein. These optimum distributions, called Tsallis distributions, are similar to Generalized Pareto Distributions, which also have an high interest in other fields, namely reliability theory \cite{Asadi2006}, climatology \cite{Montfort1986}, radar imaging \cite{LaCour2004} or actuarial sciences \cite{Cebrian2003}.

Jaynes' maximum entropy principle \cite{Jaynes1957,Jaynes1982} suggests that the least biased probability distribution that describes a partially-known system is the probability distribution with maximum entropy compatible with all the available  prior information. When prior information is available in the form of constraints on expected values, the maximum entropy method amounts to minimize Kullback information divergence $D(P||Q)$ (or equivalently maximizing Shannon $Q$-entropy) subject to normalization and these an observation constraints. In the case of a single constraint on the mean of the distribution, say $E_P[X]=m$, the minimum of Kullback information in the set of all probability distributions with expectation $m$ is of course a function of $m$, denoted $\mathcal{F}(m)$ as follows
\begin{equation}
\mathcal{F}(m)=\left\{
\begin{array}
[c]{c}%
\min_{P}  D(P||Q)\\
%\text{s.t.} \!
\begin{array}[c]{rl}
 \text{s.t.~~ } & \!\! m=E_{P}[X]\\
 \text{and } & \int_\mathcal{D} P(x) dx = 1
\end{array}\\
\end{array}~\\
\right. \hspace{-0.3cm}   \label{whole-me-pb}%
\end{equation}
It is a `contracted' version of Shannon $Q$-entropy and is called a level-1 entropy functional, or rate function, in the theory of large deviations, e.g. \cite{Ellis1985}. The maximum entropy method is a widely and successful method extensively used in a large variety of problems and contexts. \\

We focus here on solutions and properties of maximum entropy problems analog to (\ref{whole-me-pb}) for the Rényi information divergence (\ref{eq:defRenyi}), and on the associated entropy functionals. The maximum Rényi-Tsallis entropy distribution, with its power law behavior, is at the heart of nonextensive statistics, but have also be considered in \cite{Golan2002,Grendar2004a}. 
%A possible rationale can be found in \cite{Bercher2006}. 
In nonextensive statistics, one still consider the usual classical mean constraint, but also a `generalized' $\alpha$-expectation constraint. This `generalized' $\alpha$-expectation is in fact the expectation
with respect to the distribution
%Cebrian, A., Denuit, M., & Lambert, Ph. (2003). Generalized Pareto
%fit to the Society of Actuaries? large claims database. North American
%Actuarial Journal 7, 18-36.
\begin{equation}
P^{\ast}(x)=\frac{P(x)^{\alpha}Q(x)^{1-\alpha}}{\int_\mathcal{D} P(x)^{\alpha
}Q(x)^{1-\alpha}dx}, \label{P-star}%
\end{equation}
that is a weighted geometric mean of $P$ and $Q$. It is nothing else but the `escort' or zooming distribution of nonextensive statistics \cite{Tsallis2002} and multifractals. 
%The expectation with respect to $P^{\ast}$
%, noted $E^{(\alpha)}\left[X\right]$
% \begin{equation}
% E^{(\alpha)}\left[X\right]=\frac{\int_{\mathcal{D}}x P^{\ast}(x)^{\alpha}Q(x)^{1-\alpha}dx}{\int_{\mathcal{D}}P^{\ast}(x)^{\alpha}Q(x)^{1-\alpha}dx}.
% \end{equation}
%is the generalized $\alpha$ expectation, or generalized mean. 
Of course, with $\alpha=1$, the escort distribution $P^*$ reduces to $P$ and the generalized mean $E_{P^{\ast}}[X]$ reduces to the standard one.

Therefore, the maximum entropy problems associated to Rényi information divergence (\ref{eq:defRenyi}), subject to normalization and to a classical (C) or generalized (G) mean constraint states as:
 %We focus here to the solutions of the maximum entropy problems associated to Rényi information divergence (\ref{eq:defRenyi}), subject to normalization and to a classical (C) or generalized (G) mean constraint:

% \begin{equation}
% \left\{
% \begin{array}
% [c]{c}%
% \min_{P}  D_{\alpha}(P||Q)\\
% %\text{s.t.} \!
% \begin{array}[c]{rl}
%  \text{s.t.~~ } (C) & \!\! m=E_{P}[X]\\
% \text{ or } (G) & \! \! m=E_{P^{\ast}}[X]\\
%  \text{and } & \int_\mathcal{D} P(x) dx = 1
% \end{array}\\
% \end{array}~\\
% \right. \hspace{-0.3cm}   \label{whole-min-pb2}%
% \end{equation}
\begin{equation}
\mathcal{F}_{\alpha}^{(C\text{ resp. }G)}(m)
%\text{ resp. }\mathcal{F}_{\alpha}^{(\alpha )}(m)
=\left\{
\begin{array}
[c]{c}%
\min_{P}D_{\alpha}(P||Q)\\
%\text{\, s.t. } m=E_{P}[X]\text{ resp. }m=E_{P^{\ast}}[X]\\
\begin{array}[c]{rl}
 \text{s.t.~~ } (C) & \!\! m=E_{P}[X]\\
\text{ or } (G) & \! \! m=E_{P^{\ast}}[X]\\
 \text{and } & \int_\mathcal{D} P(x) dx = 1
\end{array}\\
%\text{and }  \int_\mathcal{D} P(x) dx = 1,
\end{array}
\right. \label{eq:F_def}%
\end{equation}
where $\mathcal{F}_{\alpha}^{(C)}(m)$ and $\mathcal{F}_{\alpha}^{(G)}(m)$ are the level-one entropy functionals associated to Rényi $Q$-entropy for the classical an generalized constraints respectively. Since Rényi entropy reduces to Shannon's for $\alpha=1$, functionals $\mathcal{F}_{\alpha}^{(.)}(m)$ will reduce to $\mathcal{F}(m)$ when $\alpha\rightarrow1$. 

Hence, in this paper, we consider the forms and properties of maximum entropy solutions associated to Rényi $Q$-entropy, subject to two kind of constraints, as explained above. The value of the maximum entropy problems at the optimum define classes of entropy functionals $\mathcal{F}_{\alpha}^{(.)}(m)$ associated to each choice of reference $Q$, and indexed by the parameter $\alpha$. The introduction of the reference measure $Q$, and therefore the definition of functionals $\mathcal{F}_{\alpha}^{(.)}(m)$ 
 is, to the best of our knowledge, new in this setting. In section \ref{sec:sol_min_entropy},  the exact form of the probability distributions  $P$ that realize the minimum of the Rényi information divergence in the right side of (\ref{eq:F_def}) are first derived. 
%, with originals or 
Then we give some properties of these distributions and of their partition functions. We show that the entropy functionals $\mathcal{F}_{\alpha}^{(.)}(m)$ are simply linked to these partition functions. General properties of the entropy functionals, including nonnegativity, convexity, are established. We also indicate how the problems 
(\ref{eq:F_def}) can be tackled numerically, for specific values of the constraints, even thouh the maximum entropy distributions exhibit implicit relationships. A divergence in the object space, that reduces to a Bregman divergence for $\alpha\rightarrow1$ is defined. These results are illustrated in section \ref{sec:special_cases} where we study four special cases of reference $Q$, and characterize the associated entropy functionals. It is then shown that some well-known entropies are recovered.

%%%%%%%%%%%%%%%%%%%%%%%%%%%%%%%%%%%%%%%%%%%%%%%%%%%%%%%%%%%%%%
\section{The minimum of Rényi divergence}
\label{sec:sol_min_entropy}
%%%%%%%%%%%%%%%%%%%%%%%%%%%%%%%%%%%%%%%%%%%%%%%%%%%%%%%%%%%%%%
\def\D{\mathcal{D}}

Let us define by
\begin{equation}
P_\nu(x)=\frac{\left[1+\gamma\left(x-\bar{x}\right)\right]^{\nu}}{Z_{\nu}(\gamma,\bar{x})}Q(x),
\label{eq:defPnu}
\end{equation}
a probability density function on a subset $\D$ of $\mathbb{R}$, where $\D$ ensure that the numerator of (\ref{eq:defPnu}) is always nonnegative and its integral finite. The normalization
%$\D=\D_{\gamma}\cap\D_{Q},$ where $\D_{\gamma}=\left\{ x:1+\gamma\left(x-\bar{x}\right)\geq0\right\} $ and $\D_{Q}$ is the domain of definition of $Q$.
$Z_{\nu}(\gamma,\overline{x})$ is the partition function defined by
\begin{equation}
Z_{\nu}(\gamma,\overline{x}
)=\int_{\mathcal{D}}\left[ 1+ \gamma(x-\overline{x})\right]  ^{\nu}Q(x)dx
\label{eq:defZ}
\end{equation}
The density $P_\nu$ depends of three parameters: the exponent $\nu$ which can be considered as a shape parameter, a scale parameter $\gamma$ and a location parameter $\bar{x}$. But these parameters can be also be linked. For instance, $\bar{x}$ might be a function of $\nu$ and $\gamma$. When non ambigous, we may also denote by $E_{\nu}[X]$ the statistical mean with respect to $P_\nu(x)$.\\

With these notations, we have the following result.
\begin{theorem}~

\begin{itemize}
\item[(C)] The distribution ${P}_{C}{(x)}$ in the family (\ref{eq:defPnu}) with $\nu=\xi=\frac{1}{\alpha-1}$ and $\overline{x}=E_{P}[X]=E_{\xi}[X]$, has the minimum Rényi divergence to $Q$ 
\begin{equation}
  D_\alpha(P||Q) \geq D_\alpha(P_C||Q)
\end{equation}
for all probability distributions $P(x)$ absolutely continuous with respect to $P_C(x)$  with a given (classical) expectation $\overline{x}$.
 
% In the set of all probability distributions with a given (classical) expectation, the distribution with minimum Rényi divergence to $Q$ is ${P}_{C}{(x)}$ in the family (\ref{eq:defPnu}) with $\nu=\xi=\frac{1}{\alpha-1}$ and $\overline{x}=E_{P}[X]=E_{\xi}[X]$.

% The solution to the minimization of Rényi divergence in (\ref{whole-min-pb2}) subject to the classical mean constraint (C) is ${P}_{C}{(x)}$ in the family (\ref{eq:defPnu}) with $\nu=\xi=\frac{1}{\alpha-1}$ and $\overline{x}=E_{P}[X]=E_{\xi}[X]$,
\item[(G)]  The distribution ${P}_{G}{(x)}$ in the family (\ref{eq:defPnu}) with $\nu=-\xi=\frac{1}{1-\alpha}$ and $\overline{x}=E_{P_G^*}[X]=E_{-(\xi+1)}[X]$, has the minimum Rényi divergence to $Q$ 
\begin{equation}
  D_\alpha(P||Q) \geq D_\alpha(P_G||Q)
\end{equation}
for all probability distributions $P(x)$ absolutely continuous with respect to 
$P_G(x)$ with a given generalized expectation $\overline{x}$.

%In the set of all probability distributions with a given generalized expectation, the distribution with minimum Rényi divergence to $Q$ is ${P}_{G}{(x)}$ in the family (\ref{eq:defPnu}) with $\nu=-\xi=\frac{1}{1-\alpha}$ and $\overline{x}=E_{P_G^{\ast}}[X]=E_{-(\xi+1)}[X]$.

% The solution to the minimization of Rényi divergence in (\ref{whole-min-pb2}) subject to the generalized mean constraint (G) is ${P}_{G}{(x)}$ in the family (\ref{eq:defPnu}) with
% $\nu=-\xi=\frac{1}{1-\alpha}$ and $\overline{x}=E_{P_G^{\ast}}[X]=E_{-(\xi+1)}[X]$.
\end{itemize}
\label{theo_solgen}
\end{theorem}
 
\begin{corollary}
The solution to the minimization of Rényi divergence in (\ref{eq:F_def}) is as given in theorem \ref{theo_solgen} for the particular values $\gamma^*$ of $\gamma$ such that $\overline{x}=m$.
\end{corollary}

It is important to emphasize that $\overline{x}$ is here a
statistical mean, and {not} the constraint $m$, and as such a function of $\gamma$. 
%It has been often incorrectly set to the fixed value $m$. Correct expressions in the case of Tsallis entropy maximizers can be found in \cite{Naudts2002} 

\begin{proof} See Appendix \ref{proof:theo}
\end{proof}

\begin{remark}
 \label{rem:limcase}
When $\alpha$ tends to 1, $|\nu|$ tends to $+\infty$. Let us introduce $\tilde\gamma$ such that 
$\gamma=\tilde\gamma/\nu$. Then 
\begin{equation}
P_\nu(x)=e^{\nu\log\left[1+\frac{\tilde\gamma}{\nu}\left(x-\bar{x}\right)\right] - \log{Z_{\nu}(\tilde\gamma,\bar{x})}} Q(x), 
\end{equation}
% \begin{equation}
%  P_\nu(x)=\frac{\left[1+\frac{\tilde\gamma}{\nu}\left(x-\bar{x}\right)\right]^{\nu}}{Z_{\nu}(\gamma,\bar{x})}Q(x) 
% = e^{\nu\log\left[1+\frac{\tilde\gamma}{\nu}\left(x-\bar{x}\right)\right] - \log{Z_{\nu}(\tilde\gamma,\bar{x})}} Q(x),
% \end{equation}
 and
\begin{equation}
 \lim_{|\nu|\rightarrow+\infty}P_\nu(x)= e^{{\tilde\gamma}\left(x-\bar{x}\right)- \log{Z_{\nu}(\tilde\gamma,\bar{x})}}Q(x), 
\label{eq:Plim}
\end{equation}
that is the standard exponential, which is the well-known solution of the minimisation of Kullback-Leibler divergence subject to a constraint on an expected value \cite[Theo 2.1, page 38]{Kullback1959}. In this case, the log-partition function becomes
\begin{equation}
 \label{eq:Zlim}
 \lim_{|\nu|\rightarrow+\infty} \log{Z_{\nu}(\tilde\gamma,\bar{x})} = \tilde\gamma \bar{x} - \log\int_\mathcal{D} e^{\tilde\gamma x}Q(x) dx 
\end{equation}

\end{remark}

% % The solutions to the minimization of Rényi divergence in (\ref{whole-min-pb2})
% % subject to the classical mean constraint (C) and the generalized mean
% % constraint (G) have the form
% % % \begin{equation}
% % % {P}{(x)}   =\frac{{\left[  \gamma(x-\overline {x})+1\right]  }^{\nu}}{Z_{\nu}(\gamma,\overline{x})}{Q(x)},
% % % \label{Psol}
% % % \end{equation}
% % % with \begin{itemize}
% % %       \item[(C)] exponent $\nu=\xi=\frac{1}{\alpha-1}$ and $\overline{x}=E_{P}[X]=E_{\xi}[X]$, with evident notations,
% % % \item[(G)] $\nu=-\xi$ and $\overline{x}=E_{P^{\ast}}[X]=E_{-(\xi+1)}[X]$
% % %      \end{itemize}
% % % and where
% %
% % \begin{gather}
% % (C)\ \text{\ }\ {P}_{C}{(x)}    =\frac{{\left[  \gamma(x-\overline
% % {x})+1\right]  }^{\xi}}{Z_{\xi}(\gamma,\overline{x})}{Q(x)},\\
% % \text{ with
% % }\overline{x}=E_{P_{C}}[X]=E_{\xi}[X]\label{eq:pc}\\
% % (G)\ \text{\ }P_{G}(x)    =\frac{\left(  1+\gamma(x-\overline{x})\right)
% % ^{-\xi}}{Z_{-\xi}(\gamma,\overline{x})}Q(x)\\
% % \text{ with }\overline{x}%
% % =E_{P_{G}^{\ast}}[X]=E_{-(\xi+1)}[X] \label{eq:pg}%
% % \end{gather}
% % where we used
% % \[
% % \xi=\frac{1}{\alpha-1}
% % \]
% % defined the partition function

Properties of entropy functionals $\mathcal{F}_{\alpha}^{(C)}(m)$ and $\mathcal{F}_{\alpha}^{(G)}(m)$  are of course linked to the properties of the optimum distribution (\ref{eq:defPnu}) and its partition function (\ref{eq:defZ}). In Property \ref{prop:Zsucc}, we characterize partition functions of successive exponents, which enables to derive the expression of the Rényi entropy associated to the optimum distribution. In Proposition \ref{prop:derivate}, we give the expression of the derivative of the partition function with respect to $\gamma$. Since the optimum distribution (\ref{eq:defPnu}) is `self-referential' (because it depends of its mean, which gives an implicit relation), direct determination of its parameters is difficult. It could rely on tabulation or on iterative techniques \cite{Tsallis1998}, that still suppose that the solution is an attractive fixed point. We define in Proposition \ref{prop:dualf} two functionals whose maximization provide the $\gamma$ parameter of the optimum distributions associated to the classical and generalized mean constraint. Then general properties of nonnegativity, minimum, convexity are then given in Proposition \ref{prop:Fgeneralproperties}. We also show that the two entropy optimization problems are related and that functionals $\mathcal{F}_{\alpha}^{(.)}(m)$ obey a special symmetry. Finally, we define a divergence in the space of possible means.

\begin{property}
\label{prop:Zsucc}  Partition functions of successive exponents are linked by
\begin{equation}
Z_{\nu+1}(\gamma,\overline{x})=E_{\nu+1-k}\left[  \left(  \gamma
(x-\overline{x})+1\right)  ^{k}\right]  Z_{\nu+1-k}(\gamma,\overline{x}).
\end{equation}
An interesting particular case is for k=1:
\begin{equation}
Z_{\nu+1}(\gamma,\overline{x})=E_{\nu}\left[  \gamma(x-\overline{x})+1\right]
Z_{\nu}(\gamma,\overline{x}).
\label{eqsucc}
\end{equation}
\end{property}
This is easily checked by direct calculation.
As a direct consequence, we may also observe that
$Z_{\nu+1}(\gamma,\overline{x})=Z_{\nu}(\gamma,\overline{x})$ if and only if
$\overline{x}=E_{\nu}\left[  x\right].$ When $\overline{x}$ is a fixed
parameter $m$, this will be only true for a special value $\gamma^{\ast}$ such
that $E_{\nu}\left[  x\right]=m$.

Now, using  (\ref{eqsucc}) in Property \ref{prop:Zsucc}, it is possible to give the expression of the Rényi divergence 
associated to the distribution (\ref{eq:defPnu}) and in particular to the solutions $P_C$ and $P_G$ of problems (\ref{eq:F_def}):
\begin{property} \label{prop:val_opt_renyi} The Rényi information divergence associated to the optimum distributions (\ref{eq:defPnu}) in theorem \ref{theo_solgen} is (C) $D_\alpha(P||Q) = -\log Z_{\xi}(\gamma,\overline{x})=-\log Z_{\xi+1}(\gamma,\overline{x})$, and (G)
$D_\alpha(P||Q) = -\log Z_{-\xi}(\gamma,\overline{x})=-\log Z_{-(\xi+1)}(\gamma,\overline{x})$.
\end{property}

\begin{proof}

The Rényi entropy associated to (\ref{eq:defPnu}) writes
\begin{eqnarray*}
D_{\alpha}(P||Q) & = & \frac{1}{\alpha-1}\log\int P(x)^{\alpha}Q(x)^{1-\alpha} dx\\
% & = & \frac{1}{\alpha-1}\log\int\frac{\left(1+\gamma\left(x-\bar{x}\right)\right)^{\alpha\nu}}{Z_{\nu}(\gamma,\bar{x})^{\alpha}}Q(x)dx\\
 & = & \frac{1}{\alpha-1}\log\int\left(1+\gamma\left(x-\bar{x}\right)\right)^{\alpha\nu}Q(x)dx-\frac{\alpha}{\alpha-1}\log Z_{\nu}(\gamma,\bar{x}),
\end{eqnarray*}
that simply reduces to
\[
D_{\alpha}(P||Q)=\frac{1}{\alpha-1}\log Z_{\alpha\nu}(\gamma,\bar{x})-\frac{\alpha}{\alpha-1}\log Z_{\nu}(\gamma,\bar{x}).
\]
(C) In one hand, if $\nu=\xi=\frac{1}{\alpha-1}$, then $\alpha\nu=\frac{\alpha}{\alpha-1}=\xi+1$,
and $D_{\alpha}(P||Q)=\frac{1}{\alpha-1}\log Z_{\xi+1}(\gamma,\bar{x})-\frac{\alpha}{\alpha-1}\log Z_{\xi}(\gamma,\bar{x})$
Therefore, when $\bar{x}=E_{\xi}\left[x\right]$, then (\ref{eqsucc}) gives $Z_{\xi+1}(\gamma,\bar{x})=Z_{\xi}(\gamma,\bar{x})$, and it simply remains
\[
 D_\alpha(P||Q) = -\log Z_{\xi}(\gamma,\overline{x})=-\log Z_{\xi+1}(\gamma,\overline{x}).
\]
 
(G) In the other hand, if $\nu=-\xi=\frac{1}{1-\alpha}$, then $\alpha\nu=\frac{\alpha}{1-\alpha}=-\xi-1$,
and $D_{\alpha}(P||Q)=\frac{1}{\alpha-1}\log Z_{-(\xi+1)}(\gamma,\bar{x})-\frac{\alpha}{\alpha-1}\log Z_{-\xi}(\gamma,\bar{x})$.
When $\bar{x}=E_{-(\xi+1)}\left[x\right]$, we have
$Z_{-\xi}(\gamma,\bar{x})=Z_{-(\xi+1)}(\gamma,\bar{x})$ according to  (\ref{eqsucc}) and  it remains
\[
D_\alpha(P||Q) = -\log Z_{-\xi}(\gamma,\overline{x})=-\log Z_{-(\xi+1)}(\gamma,\overline{x}).
\]
\end{proof}

Since the Rényi information divergence of distributions (\ref{eq:defPnu}) is simply the log-partition function, it will be useful to examine the behaviour of the partition function with respect to the parameter $\gamma$. Hence, the following proposition gives the expression of the derivative of the partition function.

\begin{proposition}
\label{prop:derivate} For the partition function (\ref{eq:defZ}) with domain of definition $\mathcal{D}$, the derivative with respect to $\gamma$ of the partition function with characteristic exponent $\nu$ is given by
\begin{equation}
\frac{d}{d\gamma}Z_{\nu}(\gamma,\overline{x})=\nu  \left(
E_{\nu-1}\left[  x-\overline{x}\right]  -\gamma\frac{d\overline{x}}{d\gamma
}\right)  Z_{\nu-1}(\gamma,\overline{x}). \label{eq:Zderivate}
\end{equation}
if (a) the domain $\mathcal{D}$ does not depend of $\gamma$, or (b) on subsets of $\gamma$ such that the domain increment $\delta\mathcal{D}$ associated  to the variation $\delta\gamma$ remains empty,  or (c) for $\nu>0$ in the continuous case or $\nu>1$ in the discrete case.
\end{proposition}

% \begin{property}
% \label{l5_0} The derivative of the partition function with characteristic
% exponent $\nu+1$ is given by
% \begin{equation}
% \frac{d}{d\gamma}Z_{\nu+1}(\gamma,\overline{x})=\left(  \nu+1\right)  \left(
% E_{\nu}\left[  x-\overline{x}\right]  -\gamma\frac{d\overline{x}}{d\gamma
% }\right)  Z_{\nu}(\gamma,\overline{x}). \label{eq:Zderivate}
% \end{equation}
% In particular, if $\overline{x}=E_{\nu}\left[  x\right]  ,$ then
% \begin{equation}
% \frac{d}{d\gamma}\log Z_{\nu+1}%
% (\gamma,\overline{x})=-\gamma{\left(  \nu+1\right)  }\frac{d\overline{x}}{d\gamma}.
% \label{dlog}
% \end{equation}
% \end{property}

\begin{proof} See Appendix \ref{proof:propderivate}
\end{proof}

Using this proposition on the derivative of the partition function and Property \ref{prop:Zsucc} on the link between partitions functions of succesive exponents, we readily have
\begin{property}
\label{prop:dlog}
If $\overline{x}=E_{\nu-1}\left[  X\right],$ then, with the same conditions as in proposition \ref{prop:derivate}:
\begin{equation}
\frac{d}{d\gamma}\log Z_{\nu}%
(\gamma,\overline{x})=-\gamma \nu \frac{d\overline{x}}{d\gamma},
\label{eq:dlogdgamma}
\end{equation}
and
\begin{equation}
\frac{d}{d\overline{x}}\log Z_{\nu}%
(\gamma,\overline{x})=-\gamma \nu.
\label{eq:dlogdx}
\end{equation}
\end{property}
This is immediately checked using (\ref{eqsucc}) and  (\ref{eq:Zderivate}) with $\overline{x}=E_{\nu-1}[X]$.  It is now interesting to consider the special case where $\overline{x}$ is a fixed value, say $m$. Then, it is immediate to check that the extrema of the function $\log Z_\nu(\gamma,m)$ occur for $\gamma^*$ such that $m=E_{\nu-1}[X]$:
\begin{property}
\label{prop:maxlogZ}
If $\overline{x}$ is a fixed value $m$, then
\begin{equation}
\left. \frac{d}{d\gamma}\log Z_{\nu}%
(\gamma,m)\right|_{\gamma=\gamma^*}=0.
\end{equation}
if and only if $\gamma^*$ is such that $m=E_{\nu-1}\left[  X\right]$.
\end{property}

This result is important because it provides an easy way to find the value of the parameter $\gamma$ of the optimum distributions (\ref{eq:defPnu}) that solves the maximum entropy problems (\ref{eq:F_def}).
\begin{proposition} 
\label{prop:dualf}
The values $\gamma^*$ of the parameter $\gamma$ of the optimum distributions that solve the maximum entropy problems (\ref{eq:F_def}) are the minimum of the maximizers  of
\begin{align}
 D_C(\gamma) &=-\log Z_{\xi+1}(\gamma,m) \label{eq:Dc} \\
 D_G(\gamma)& =-\log Z_{-\xi}(\gamma,m) \label{eq:Dg}
\end{align}
where the two partitions functions involved are convex, possibly on several well defined intervals. Then, the entropy functionals $\mathcal{F}_\alpha^{(.)}$ are simply given by
\begin{equation}
\mathcal{F}_\alpha^{(C\text{~resp.~} G)}(m) = D_{C \text{~resp.~} G}(\gamma^*).
\end{equation}
\end{proposition}
\begin{proof}
 Indeed, Theorem \ref{theo_solgen} and its corrolary indicates that the solution for the classical constraint (C) is obtained for $\overline{x}=m=E_{\xi}\left[  X\right]$ and by $\overline{x}=m=E_{-\xi-1}\left[  X\right]$ for the generalized constraint (G). Then by Property \ref{prop:maxlogZ} it suffices to look for the extrema of $D_C(\gamma)=-\log Z_{\xi+1}(\gamma,m)$ in the first case or of $D_G(\gamma)=-\log Z_{-\xi}(\gamma,m)$ in the second case.
With similar conditions of  derivation as in Proposition \ref{prop:derivate} the second derivative of the partition function with respect to $\gamma$ writes
\begin{align}
\frac{d^2 Z_\nu(\gamma,m)}{d\gamma^2} & = \nu(\nu-1)\int_\mathcal{D} \left( x-m\right)^2
\left[1+\gamma\left(x-\bar{x}_{\gamma}\right)\right]^{\nu-2} Q(x) dx \\
& = \nu(\nu-1) E_{\nu-2}\left[(X-m)^2\right] Z_{\nu-2}(\gamma,m).
\end{align}
For $\nu=\xi+1$ and $\nu=-\xi$, the factor $\nu(\nu-1)$ reduces to $\frac{\alpha}{(\alpha-1)^2}$. Since $\alpha$ is positive, the second derivative is always positive and the partition functions $Z_{\xi+1}(\gamma,m)$ and $Z_{-\xi}(\gamma,m)$ are convex on their domain of definition. On these domains, the functionals in (\ref{eq:Dc}) and (\ref{eq:Dg}) are then unimodal and their extrema are maxima.

In the discrete case and for $\nu<0$, $Z_\nu(\gamma,m)$ has singularities for all $\gamma=\frac{1}{m-k}$, where $k$ is an integer in the support of the distribution. Therefore,  $Z_\nu(\gamma,m)$ is only defined on segments $\left( \frac{1}{m-k}, \frac{1}{m-k-1} \right)$, for $m \not \in (k+1,k)$), and $\left( \frac{1}{m-k-1}, \frac{1}{m-k} \right)$ for $m \in (k+1,k)$. In such a case, $-\log Z_\nu(\gamma,m)$ may present several maxima. The situation $\nu<0$ occurs for the classical constraint when $\alpha\in (0,1)$ (since the index $\xi+1=\alpha/(\alpha-1)$ is negative), and for the generalized constraint when $\alpha>1$. An example of functional $D_C(\gamma)$ with $\alpha=0.5$ in the case of a Poisson distribution is reported in Fig.~\ref{Fig_Pois_dual_exampl}. In the $\nu>0$ discrete case 
%(generalized constraint)
or in the continuous case, there is a single maximum.

Finally, since the expression of the Rényi information divergence of the optimum distributions is precisely the opposite of the log-partition function as indicated in Property \ref{prop:val_opt_renyi}, the value of functionals (\ref{eq:Dc}) and (\ref{eq:Dg}) at their optima $\gamma^*$ such that $\bar{x}=m$ is precisely the value of entropy functionals $\mathcal{F}_\alpha^{(1)}(m)$ and $\mathcal{F}_\alpha^{(\alpha)}(m)$.
\end{proof}

\begin{remark}
 When $\alpha$ tends to 1, the parameter $\tilde\gamma^*$ is thus the maximizer of (\ref{eq:Zlim}), and we obtain
\begin{equation}
\lim_{\alpha\rightarrow1}\mathcal{F}_\alpha^{(.)} = \sup_{\tilde\gamma}\left\{ \tilde\gamma \bar{x} - \log\int_\mathcal{D} e^{\tilde\gamma x}Q(x) dx \right\},
\end{equation}
 that is the Cramér transform of $Q(x)$.
\end{remark}

With the help of these different results it is now possible to characterize more precisely the entropy functionals 
\begin{proposition}
\label{prop:Fgeneralproperties} Entropy functionals $\mathcal{F}_{\alpha}^{(C)}(m)$ and
$\mathcal{F}_{\alpha}^{(G)}(m)$ are nonnegative, with an unique minimum
at $m_{Q\text{ }}$, the mean of $Q,$ and $\mathcal{F}_{\alpha}^{(.)}%
(m_{Q})=0.$ Furthermore, $\mathcal{F}_{\alpha}^{(C)}(m)$ is strictly convex
for $\alpha\in\lbrack0,1].$
\end{proposition}

\begin{proof}
Rényi information divergence $D_{\alpha}(P||Q)$ is always nonnegative, and
equal to zero for $P=Q$. Since functionals $\mathcal{F}_{\alpha}^{(.)}(x)$ are
defined as the minimum of $D_{\alpha}(P||Q)$, they are always nonnegative.
If $P=Q,$ we have also $P^{\ast}=Q$ and $m=E_{P}[X]$ =$E_{P^{\ast}%
}[X]=m_{Q}.$ Therefore $\mathcal{F}_{\alpha}^{(.)}(m_{Q})=0$ and $m_{Q}$ is a
global minimum. \newline
From (\ref{eq:dlogdx}), we have $\frac{d}{d\overline{x}}\log Z_{\nu+1}%
(\gamma,\overline{x})=-\gamma{\left(  \nu+1\right) }$. Then, functionals $\mathcal{F}_{\alpha}^{(.)}(x)$
are only minimum if $\gamma=0$, and the corresponding optimum probability distributions are simply $P=Q,$ and $D_{\alpha
}(Q||Q)=0.$ Therefore, $\mathcal{F}_{\alpha}^{(.)}(x)$  have an unique minimum for $x=m_{Q},$ the mean of
$Q$, and $\mathcal{F}_{\alpha}^{(.)}(m_{Q})=0.$ \newline
Finally, we examine the convexity of $\mathcal{F}_{\alpha}^{(C)}(m),$ for $\alpha\in\lbrack0,1]$.\par
Let $P_1$ and $P_2$ be the distributions that achieve the
minimization of $D_{\alpha}(P||Q)$ subject to the constraints $x_{1}=$
$E_{P}[X]$ and $x_{2}=E_{P}[X]$ respectively. Then, $\mathcal{F}%
_{\alpha}^{(C)}(x_{1})=D_{\alpha}(P_1||Q),$ and $\mathcal{F}_{\alpha
}^{(C)}(x_{2})=D_{\alpha}(P_2||Q)$. In the same way, denote
$\mathcal{F}_{\alpha}^{(C)}(\mu x_{1}+(1-\mu)x_{2})=D_{\alpha}(\hat{P}%
||Q),$ where $\hat{P}$ denotes the optimum distribution with mean $\mu
x_{1}+(1-\mu)x_{2}.$ Distributions $\hat{P}(u)$ and $\mu P_1(u)+(1-\mu)P_2(u)$ have the same mean $\mu x_{1}+(1-\mu)x_{2}$.
Hence, when $D_{\alpha}(P||Q)$ is a convex function of $P,$ that is for
$\alpha\in\lbrack0,1],$ \ we have $D_{\alpha}(P^{\ast}||Q)\leq\mu
D_{\alpha}(P_1(u)||Q)+(1-\mu)D_{\alpha}(P_2(u)||Q),$ that is
$\mathcal{F}_{\alpha}^{(C)}(\mu x_{1}+(1-\mu)x_{2})\leq\mu\mathcal{F}_{\alpha
}^{(C)}(x_{1})+(1-\mu)\mathcal{F}_{\alpha}^{(C)}(x_{2})$ and $\mathcal{F}%
_{\alpha}^{(C)}(x)$ is a convex function.
\end{proof}

Up to now the two optimization problems have been considered in parallel.
But here is a special symmetry that enables to relate the solutions
of the minimization of Rényi divergence subject to classical and generalized mean constraints. Then, there exists a simple relationship between the entropy functionals
$\mathcal{F}_{\alpha}^{(C)}(x)$ and $\mathcal{F}_{\alpha}^{(G)}(x)$.
%*******************

Let us consider our original Rényi divergence minimization problem, on one hand with index $\alpha_1$ and subject to a classical mean constraint $m$, and on the other hand with index $\alpha_2$ and subject to a generalized mean constraint $m$.
The associated functionals, by Property \ref{prop:dualf}, are $D_C(\gamma)=-\log Z_{\xi_{1}%
+1}(\gamma,m)$ and $D_G(\gamma)=-\log Z_{-\xi_{2}}(\gamma,m).$ Thus, we will have
pointwise equality of these functions if $\xi_{1}+1=-\xi_{2},$ that is if
indexes $\alpha_{1}$ and $\alpha_{2}$ satisfy $\alpha_{1}=1/\alpha_{2}.$ In
this case, we will of course have equality of the optimum parameters $\gamma,$ and the
two optimization problems will have the same optimum value.
Because of the pointwise equality functions $D_G(\gamma)$ and $D_G(\gamma)$, it is clear that the associated divergences are equal at the optimum, that is $D_{\alpha_{1}}%
(P_{C}||Q)=D_{\alpha_{2}}(P_{G}||Q).$ Besides this is easily checked in the general case: for the escort
distribution $P^{\ast}(x)$ in (\ref{P-star}),  we \emph{always} have the equality $D_{\frac{1}{\alpha}}(P^{\ast
}||Q)=D_{\alpha}(P_{1}||Q)$. Hence, the minimization of the
$\alpha$ Rényi divergence subject to the generalized mean constraint is
exactly equivalent to the minimization of the $1/\alpha$ Rényi divergence
subject to the classical mean constraint
\begin{equation}
\left\{
\begin{array}
[c]{c}%
\inf_{P_{1}}D_{\alpha}(P_{1}||Q)\\
s.t\text{ \ }E_{P^{\ast}}\left[  X\right]  =m
\end{array}
\right.  =\left\{
\begin{array}
[c]{c}%
\inf_{P^{\ast}}D_{\frac{1}{\alpha}}(P^{\ast}||Q)\\
s.t\text{ \ }E_{P^{\ast}}\left[  X\right]  =m
\end{array}
\right.  , \label{eq:equality_pbs_1_over_alpha}%
\end{equation}
so that generalized and classical mean constraints can always be swapped,
provided the index $\alpha$ is changed into $1/\alpha$, as was argued in
\cite{Raggio1999a,Naudts2002}. Hence, equality (\ref{eq:equality_pbs_1_over_alpha}) enables us to complete the characterization of entropy functionals $\mathcal{F}_{\alpha}^{(C)}(m)$ and $\mathcal{F}_{\alpha}^{(G)}(m)$:
\begin{property}
\label{prop:duality}
Entropy functionals $\mathcal{F}_{\alpha}^{(C)}(m)$ and
$\mathcal{F}_{\alpha}^{(G)}(m)$
%are nonnegative, with an unique minimum
%at $m=m_{Q},$ the mean of $Q,$ and obey to
admit
the symmetry $\mathcal{F}_{\alpha}^{(G)}(x)=\mathcal{F}%
_{1/\alpha}^{(C)}(x).$ Besides, $\mathcal{F}_{\alpha}^{(C)}(m)$ is strictly
convex for $\alpha\in\lbrack0,1]$ and $\mathcal{F}_{\alpha}^{(G)}(m)$ is
strictly convex for $\alpha\in\lbrack1,+\infty].$
\end{property}

Interestingly, it is also possible to define a \emph{divergence in the object space}, that is a kind of generalized distance between two ``objects''. These divergences may be used for instance in clustering \cite{Nock2006}. The objects are here considered as generalized means of distributions with minimum divergence to a reference measure $Q(x)$.

\begin{proposition}
If $P_1$ and $P_2$ are two distributions in (\ref{eq:defPnu}) with exponent $\nu=-\xi$ (generalized constraint),  with $P_2\ll P_1$, and with respective parameters $\gamma_1$, $\gamma_2$ and means $m_1$, $m_2$, then
\begin{align}
& \mathcal{F}_{\alpha}^{(G)}(m_2,m_1) = D_{\alpha}(P_{2}||P_1)
 = \mathcal{F}_{\alpha}^{(G)}(m_2)  - \mathcal{F}_{\alpha}^{(G)}(m_1) \nonumber \\& + \frac{1}{\alpha-1}\log\left(
1-(\alpha-1)\frac{d \mathcal{F}_{\alpha}^{(G)}}{dm}(m_1)(m_2-m_1)
\right),
\label{eq:bregman-like}
\end{align}
and $\mathcal{F}_{\alpha}^{(G)}(m_2,m_1) \ge 0$, with equality if and only if $m_2=m_1$.
\end{proposition}
\begin{proof}
The result is obtained by simple computations. First, we have
\[
D_\alpha(P_2||P_1) = \frac{1}{\alpha-1} \log{\int\frac{[1+\gamma_2(x-m_2)]^{\frac{\alpha}{1-\alpha}} }{Z_{-\xi}(\gamma_2,m_2)^\alpha}\frac{[1+\gamma_1(x-m_1)]}{Z_{-\xi}(\gamma_1,m_1)^{1-\alpha}} Q(x) dx}
\]
which can be rewritten as
\begin{align}
D_\alpha(P_2||P_1) = & \frac{1}{1-\alpha}\left( \alpha\log Z_{-\xi}(\gamma_2,m_2) + (1-\alpha)\log Z_{-\xi}(\gamma_1,m_1)
- \log Z_{-\xi-1}(\gamma_2,m_2) \right) \\
+ &  \frac{1}{\alpha-1} \log \left[  1 + \gamma_1 \int (x-m_1) \frac{[1+\gamma_2(x-m_2)]^{-\xi-1}}{Z_{-\xi-1}(\gamma_2,m_2)} Q(x) dx \right]
\end{align}
In the first line, we have $Z_{-(\xi+1)}(\gamma_2,m_2)=Z_{-\xi}(\gamma_2,m_2)$ by Property \ref{prop:Zsucc}, eq. (\ref{eqsucc}), and we recognize from Proposition \ref{prop:dualf} that $\mathcal{F}_{\alpha}^{(G)}(m) = -\log Z_{-\xi}(\gamma,m)$. In the second line, the integral reduces to $(m_2-m_1)$ since $m_2$ is the generalized mean of the distribution $P_2$. Finally, $\gamma_1$ can be expressed as the derivative of the log-partition function as stated by (\ref{eq:dlogdx}) in Property \ref{prop:dlog}.

By definition, $\mathcal{F}_{\alpha}^{(G)}(m_2,m_1)$ is the Rényi information divergence  $D_{\alpha}(P_{2}||P_1)$ which is always greater or equal to zero, with equality if and only if $P_2=P_1$, which implies $m_2=m_1$.
\end{proof}
%The ingredients of the proof are to recognize that $m_2=E_{-(\xi+1)}[X]$, to note that $\alpha\xi=\xi+1$ and that $Z_{-(\xi+1)}(\gamma_2,m_2)=Z_{-\xi}(\gamma_2,m_2)$ by Property 1, eq. (\ref{eqsucc}), and to express $\gamma_1$ using (\ref{dlog}) in Property 2.

For $\alpha\rightarrow 1$, $\mathcal{F}_{\alpha}^{(G)}(m_2,m_1)$ reduces to a standard \textit{Bregman divergence}. Indeed, using $\log(1-x)\simeq -x$, we have simply
\[
\lim_{\alpha\rightarrow1} \mathcal{F}_{\alpha}^{(G)}(m_2,m_1)
 = \mathcal{F}_{\alpha}^{(G)}(m_2)  - \mathcal{F}_{\alpha}^{(G)}(m_1)
- \frac{d \mathcal{F}_{\alpha}^{(G)}}{dm}(m_1)(m_2-m_1).
\]
%, that is $|\xi|\rightarrow+\infty$,

\section{Examples of entropy functionals}
\label{sec:special_cases}

We now examine 4 special cases for the reference mesure $Q(x)$:
a uniform and an exponential distribution that model systems with continuous states; and then a Bernoulli (two-levels) and a Poisson distribution which may model systems with discrete states.  
% In each case, we work out the partition and dual functions, and study
% entropies $\mathcal{F}_{\alpha}^{(1\text{ or }\alpha)}(x)$ for $\alpha\rightarrow 1$. 
The minima of the Rényi divergence, that is the  entropies $\mathcal{F}_{\alpha}^{(C\text{ or }G)}(x)$, are attained for the values $\gamma^{\ast}$ that maximize the functionals $D_C(\gamma)$ and $D_G(\gamma)$ in Proposition \ref{prop:dualf}. This involves the computation of $Z_{\nu}(\gamma,m)$ for all reference measures $Q$ considered, and the resolution of $\frac{d}{d\gamma}Z_{\nu+1}(\gamma,m)=0$.
The case $\alpha=1$  is obtained in the limit  $|\nu|\rightarrow +\infty$, since $|\xi|\rightarrow +\infty$ when $\alpha$ tends to 1. Results of numerical evaluations for varying $\alpha$ are provided.

%%%%%%%%%%%%%%%%%%%%%%%%%%%%%%%%%%%%%%%%%%%%%%%%%%%%%%%%%%%%%%%
\subsection{Uniform reference}

Let us first consider the case of the uniform reference $Q(x)$ on $[0,1$. The partition function is given by $Z_{\nu}%
(\gamma,m)=\int_{\mathcal{D}}\left[  \gamma(x-m)+1\right]  ^{\nu}dx,$ where
the domain $\mathcal{D}$ is defined by $\mathcal{D=D}_{Q}\cap\mathcal{D}%
_{\gamma},$ with $\mathcal{D}_{Q}=\left\{
x:x\in\lbrack0,1]\right\}  $ and $\mathcal{D}_{\gamma}=\left\{  x:\gamma
(x-m)+1\geq0\right\}$. 

%reduces to 
% \[
% \mathcal{D}_{\gamma}=\left\{  x\in\lbrack
% m-\frac{1}{\gamma},+\infty\lbrack\text{ if }\gamma\geq0,x\in\lbrack
% 0,m-\frac{1}{\gamma}]\text{ if }\gamma\leq0\right\}. 
% \]
Computation of the partition function in the different domains together with the fact that $m\in\lbrack0,1]$ leads to
\begin{align*}
Z_{\nu}(\gamma,m)  &  =\allowbreak\frac{1}{\gamma\left(  1+\nu\right)
}\left(  \left(  \gamma-\gamma m+1\right)  ^{\nu+1}U(\gamma-\frac{1}%
{m-1})-\left(  -\gamma m+1\right)  ^{\nu+1}U(-\gamma+\frac{1}{m})\right)  ,\\
\text{for all }\gamma\text{ if }\nu &  \geq0,\text{ for }\gamma\in\left(
\frac{1}{m-1},\frac{1}{m}\right)  \ \ \text{if }\ \nu<0,\text{ and }Z_{\nu
}(\gamma,m)=\allowbreak+\infty\text{ otherwise},
\end{align*}
where $U$ denotes the Heaviside distribution: $U(t)=0$ for $t<0$ and $U(t)=1$
for $t>0$.

The first derivative of the partition function is given by%
\begin{equation}
\frac{d}{d\gamma}Z_{\nu}(\gamma,m)=-\frac{\nu\gamma\left(  m-1\right)
+1}{\gamma^{2}\left(  \nu+1\right)  }\left(  \gamma\left(  m-1\right)
+1\right)  ^{\nu}U(\gamma-\frac{1}{m-1})+\frac{\gamma m(\nu)+1}{\gamma
^{2}\left(  \nu+1\right)  }\left(  1-\gamma m\right)  ^{\nu}U(-\gamma+\frac
{1}{m}).
\end{equation}
% Then, the alternate dual functions are explicitly given by
% \begin{align}
% \widetilde{D}_{C}(\gamma)  &  =-\log Z_{\xi+1}(\gamma,m)=-\log\left(
% \frac{\left(  \gamma-\gamma m+1\right)  ^{\xi+2}U(\gamma-\frac{1}%
% {m-1})-\left(  -\gamma m+1\right)  ^{\xi+2}U(-\gamma+\frac{1}{m})}%
% {\gamma\left(  2+\xi\right)  }\right) \label{eq:dc_uni}\\
% &  \text{for the classical mean constraint, and}\nonumber\\
% \widetilde{D}_{G}(\gamma)  &  =-\log Z_{-\xi}(\gamma,m)=-\log\left(
% \frac{\left(  \gamma-\gamma m+1\right)  ^{1-\xi}U(\gamma-\frac{1}%
% {m-1})-\left(  -\gamma m+1\right)  ^{1-\xi}U(-\gamma+\frac{1}{m})}%
% {\gamma\left(  1-\xi\right)  }\right) \label{eq:dg_uni}\\
% &  \text{for the generalized mean constraint.}\nonumber
% \end{align}
% 
% Entropy functionals $\mathcal{F}_{\alpha}^{(C)}(x)$ and $\mathcal{F}_{\alpha
% }^{(\alpha)}(x)$ are defined as the minimum of $D_{\alpha}(P||Q),$ subject to
% the mean constraint. 
We next have to look for the expression of entropy functionals $\mathcal{F}_{\alpha}^{(.)}(x)$.
Unfortunately, no analytical solution can be exhibited here, but the two functionals still can be evaluated numerically. For the classical mean constraint (C) we can check that $\mathcal{F}_{\alpha}^{(C)}(x)$ is a family of convex functions on $(0,1)$, minimum for the mean of the reference measure $Q$, as was indicated in Proposition \ref{prop:Fgeneralproperties}. 
In the same way, we can check that for the generalized mean constraint (G)
$\mathcal{F}_{\alpha}^{(G)}(x)$ is a family of nonnegative functions on $(0,1)$, also
minimum for the mean of the reference measure $Q$. The entropies $\mathcal{F}_{\alpha}^{(C)}(x)$ and $\mathcal{F}_{\alpha}^{(G)}(x)$ were evaluated numerically and are given in Figs. \ref{Fig_Uni_classical} and \ref{Fig_Uni_generalized} for $\alpha\in(0,1)$. Of course, the $\alpha\leftrightarrow1/\alpha$ duality given in Property \ref{prop:duality} enables to extend these two functionals for $\alpha>1$.

% %%%%%% FIGURE 2 %%%%%%%%%%%%%%%%%%%%%%%%%%%%%%%%%%%%%%%%%
% \begin{figure}[ptbh]
% \begin{center}
% \includegraphics[width=4in,height=2.8in]{fig2.pdf}
% \end{center}
% \caption{Entropy functional $\mathcal{F}_{\alpha}^{(C)}(x)$ \ for a uniform
% reference measure and $\alpha\in[0,1]$.}%
% \label{Fig_Uni_classical}%
% \end{figure}
% %%%%%%%%%%%%%%%%%%%%%%%%%%%%%%%%%%%%%%%%%%%%%%%%%%%%%%%%%
% 
% 
% %%%%%% FIGURE 3 %%%%%%%%%%%%%%%%%%%%%%%
% \begin{figure}[ptbh]
% \begin{center}
% \includegraphics[width=4in,height=2.8in]{fig3.pdf}
% \end{center}
% \caption{Entropy functional $\mathcal{F}_{\alpha}^{(G)}(x)$ \ for a
% uniform reference measure and $\alpha\in[0,1]$.}%
% \label{Fig_Uni_generalized}%
% \end{figure}
% %%%%%%%%%%%%%%%%%

\begin{figure}
 \begin{minipage}[b]{.46\linewidth}
  \centering\includegraphics[width=3in]{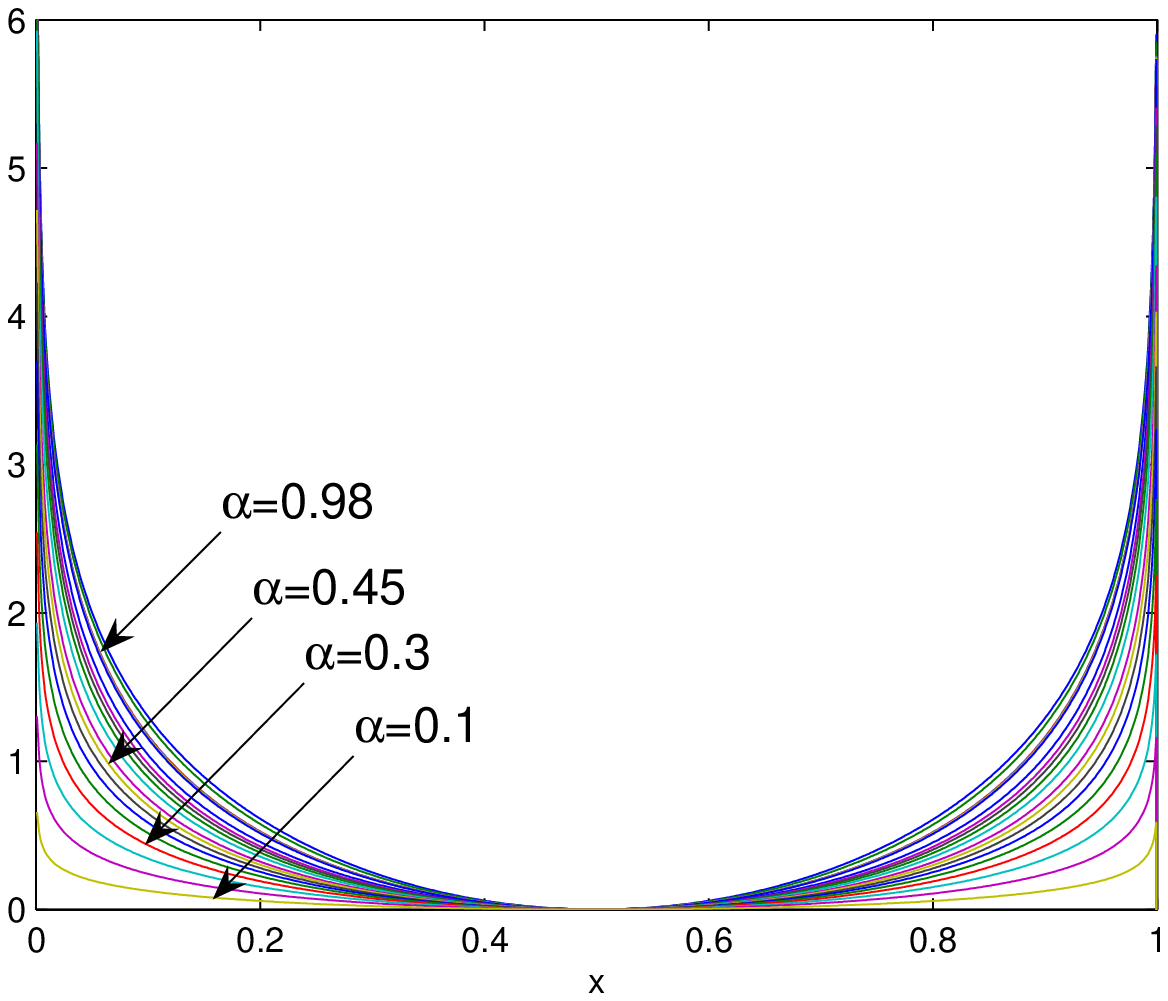}
\caption{Entropy functional $\mathcal{F}_{\alpha}^{(C)}(x)$ \ for a uniform
reference measure and $\alpha\in(0,1)$.}%
\label{Fig_Uni_classical}%
 \end{minipage} \hfill
 \begin{minipage}[b]{.46\linewidth}
  \centering \includegraphics[width=3in]{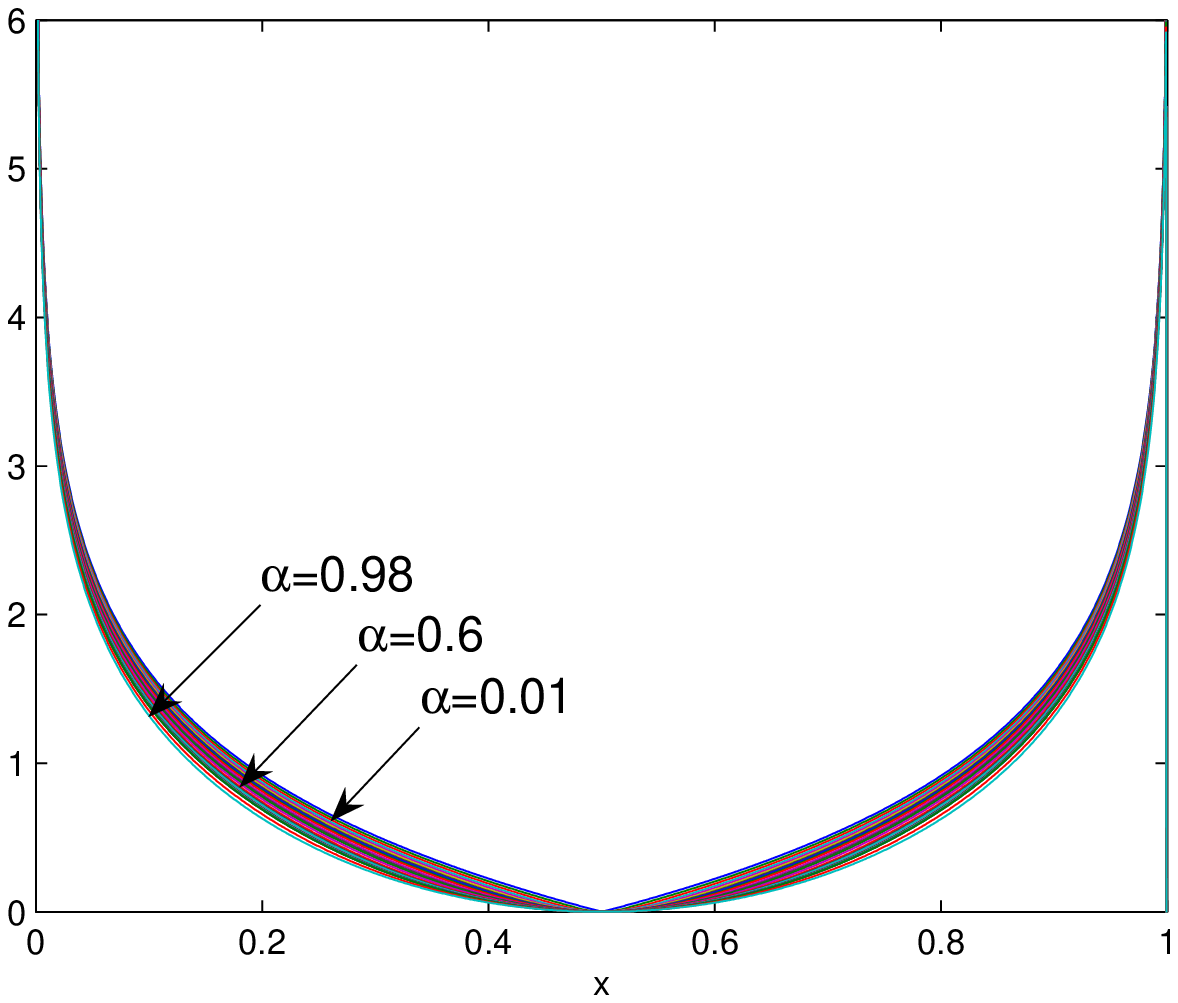}
\caption{Entropy functional $\mathcal{F}_{\alpha}^{(G)}(x)$ \ for a
uniform reference measure and $\alpha\in(0,1)$.}%
\label{Fig_Uni_generalized}%
 \end{minipage}
\end{figure}

Hence, it is apparent that the minimization of $\mathcal{F}_{\alpha}^{(.)}(x)$ 
under some constraint would automatically lead to a solution on (0,1). Moreover, the
parameter $\alpha$ may serve to tune the curvature of the functional and the
degree of penalization of bounds.

\subsection{Exponential reference}

The exponential probability density function is $Q(x)=\beta e^{-\beta x},$ for $x \geq 0$ and $\beta>0$. 
%  in $\mathcal{D}_{Q}=\left\{  x:x\geq0\text{ }\right\}$.  
% for
% $x\geq0$ and $\beta>0$. The domain $\mathcal{D}$ is defined by $\mathcal{D=D}%
% _{Q}\cap\mathcal{D}_{\gamma},$ where $\mathcal{D}_{Q}$ is now $\mathcal{D}%
% _{Q}=\left\{  x:x\geq0\text{ }\right\}  $ and $\mathcal{D}_{\gamma}=\left\{
% x:\gamma(x-m)+1\geq0\right\}  =\left\{  x:x\geq\frac{\gamma m-1}{\gamma}\text{
% if }\gamma>0\text{ or }x\leq\frac{\gamma m-1}{\gamma}\text{ if }\gamma<0\text{
% }\right\}  $.
The partition function is given by%
\begin{equation}
Z_{\nu}(\gamma,m)=\beta\int_{\mathcal{D}}\left[  \gamma(x-m)+1\right]  ^{\nu
}e^{-\beta x}dx
\end{equation}
where $\mathcal{D}=\left\{
x:x\geq\max\left\{  0,m-\frac{1}{\gamma}\right\}  \text{ if }\gamma>0\text{ or
}x\in\lbrack0,m-\frac{1}{\gamma}]\text{ if }\gamma<0\text{ }\right\}$, ensures  that the integrand $\left[\gamma(x-m)+1 \right]$ is nonnegative and the integral finite.

The evaluation of $Z_{\nu}(\gamma,m)$ on the different domains gives: 
%(details of the derivations can be found in \cite{BercherWeb}):%
\begin{equation}
Z_{\nu}(\gamma,m)=\left\{
\begin{array}
[c]{ccc}%
e^{-\beta\frac{\gamma m-1}{\gamma}}\left(  \frac{\gamma}{\beta}\right)  ^{\nu
}\Gamma\left(  \nu+1\right)  & \text{if } &
\begin{array}
[c]{c}%
\gamma>\frac{1}{m}>0\\
\nu\geq0
\end{array}
\\
e^{-\beta\frac{\gamma m-1}{\gamma}}\left(  \frac{\gamma}{\beta}\right)  ^{\nu
}\Gamma(\nu+1,\beta\frac{1-\gamma m}{\gamma}) & \text{if } & \frac{1}%
{m}>\gamma>0\\
e^{-\beta\frac{\gamma m-1}{\gamma}}\left(  \frac{\beta}{\gamma}\right)
^{-\nu}\left(  \Gamma\left(  \nu+1,\beta\frac{-\gamma m+1}{\gamma}\right)
-\Gamma\left(  \nu+1\right)  \right)  & \text{if} &
\begin{array}
[c]{c}%
\gamma<0<\frac{1}{m}\\
\nu\geq0
\end{array}
\end{array}
\right.
\end{equation}
and $Z_{\nu}(\gamma,m)=+\infty$ for $\gamma<0$ or $\gamma>\frac{1}{m}$ if
$\nu<0.$
% 
% As for the derivative of the partition function, we obtain%
% \begin{equation}
% \frac{d}{d\gamma}Z_{\nu}(\gamma,m)=\left\{
% \begin{array}
% [c]{c}%
% e^{-\beta\frac{\gamma m-1}{\gamma}}\left(  \frac{\gamma}{\beta}\right)  ^{\nu
% }\frac{-\beta+\nu\gamma}{\gamma^{2}}\Gamma\left(  \nu+1\right) \\
% \frac{1}{\gamma^{2}}\left(  \frac{\gamma}{\beta}\right)  ^{\nu}\left(
% e^{-\beta\frac{\gamma m-1}{\gamma}}\left(  \nu\gamma-\beta\right)
% \Gamma\left(  \nu+1,-\beta\frac{\gamma m-1}{\gamma}\right)  \allowbreak
% +\left(  -\beta\frac{\gamma m-1}{\gamma}\right)  ^{\nu}\beta\right) \\
% \frac{1}{\gamma^{2}}\left(  \frac{\beta}{\gamma}\right)  ^{-\nu}\left(
% e^{-\beta\frac{\gamma m-1}{\gamma}}\left(  \Gamma\left(  \nu+1,-\beta
% \frac{\gamma m-1}{\gamma}\right)  -\Gamma\left(  \nu+1\right)  \right)
% \allowbreak\left(  -\beta+\nu\gamma\right)  +\left(  -\beta\frac{\gamma
% m-1}{\gamma}\right)  ^{\nu}\beta\right)
% \end{array}
% \right.
% \end{equation}
% in the same domains as before. 
% % 
% Then, the alternate dual functions for the
% classical and generalized mean constraints are simply given by
% (\ref{eq:dual_functions_for_spcases}) and their gradients are easily expressed
% using the expressions of the partition function and of its derivative.

%% \subsubsection{The limit case $\alpha\rightarrow1$}

Let us now examine the behavior of the entropies $\mathcal{F}_{\alpha}%
^{(.)}(x)$ when $\alpha\rightarrow1$. This amounts to study 
$Z_{\nu}(\gamma,m)$ and its maximum when $\left\vert \nu\right\vert \rightarrow+\infty$. 
%Indeed, for the classical mean constraint, $\nu=\xi+1$
%$\rightarrow-\infty$ and for the generalized mean $\nu=-\xi
%\rightarrow+\infty$.

The simplest derivation is as follows. 
As in Remark \ref{rem:limcase}, let $\gamma=\tilde{\gamma}/\nu$, so that $(1+\gamma(x-m))^\nu \sim \exp(\tilde{\gamma}(x-m))$. In this case, one easily obtain that 
\begin{equation}
\log Z_{\nu}(\tilde{\gamma},m)\simeq
\log\beta-\tilde{\gamma} m - \log(\beta-\tilde{\gamma}),
\end{equation}
whose derivative is equal to zero for
\begin{equation}
\tilde{\gamma}^{\ast}=\beta-\frac{1}{m}.%
\end{equation}
We shall  also note that if $\nu<0,$ the sign of $\gamma=\tilde{\gamma}/\nu$ is the sign of $\left(
1-\beta m\right).$ Since $Z_{\nu}(\gamma,m)$ is only defined for $\gamma>0$ when $\nu<0,$ it means that we only have a solution for $m<1/\beta$. Indeed, for $\gamma>0$ and $\nu<0$, the factor
$(1+\gamma(x-m))^\nu$ is decreasing, and consequently the mean of the optimum distribution (\ref{eq:defPnu}) cannot be greater than the mean of the reference distribution, $m_Q=E_Q[X]=1/\beta$. 

With the optimum value $\tilde{\gamma}^{\ast}$, the log partition function becomes%
\begin{equation}
\log Z_{\nu}(\gamma^{\ast},m)\simeq-\left(  \beta m-1\right)  +\log\left(
\beta m\right)  \ \ \ \ \ (\forall m\text{ if }\nu\rightarrow+\infty,\text{for
}m<1/\beta\text{ if }\nu\rightarrow-\infty).
\end{equation}
Finally, we thus obtain%
\begin{equation}
\mathcal{F}_{\alpha\rightarrow1}^{(C)}(x)  =  -\log Z_{\xi+1}(\gamma^{\ast},x)=\left(  \beta x-1\right)  -\log\left(  \beta x\right),
\label{eq:itakura}
\end{equation}
for $x<1/\beta$ when $\alpha$ tends to 1 by lower values, and for all $x$ if $\alpha$ tends to 1 by higher values. 
By the duality property \ref{prop:duality}, this expression is also the limit form of functional $\mathcal{F}_{\alpha}^{(G)}(x)$.
% 
% \begin{equation}
% \left\{
% \begin{array}
% [c]{lcl}%
% \mathcal{F}_{\alpha\rightarrow1}^{(1)}(x) & = & -\log Z_{\xi+1}(\gamma^{\ast
% },x)=\left(  \beta x-1\right)  -\log\left(  \beta x\right)
% \ \ \ \ \ \ \ \ \text{for }m<1/\beta\\
% \mathcal{F}_{\alpha\rightarrow1}^{(\alpha)}(x) & = & -\log Z_{-\xi}%
% (\gamma^{\ast},x)=\left(  \beta x-1\right)  -\log\left(  \beta x\right)
% \end{array}
% \right.
% \end{equation}

As was expected, the functional $\left(  \beta x-1\right)  -\log\left(  \beta
x\right)  $ is strictly convex, positive and zero for $x=1/\beta,$ the mean of
the exponential distribution. It was employed in speech processing and is
called the \emph{Itakura-SaÃ¯to entropy functional. }For $\beta=1,$ it reduces
to the so-called \emph{\ Burg entropy} that is well-known in spectrum analysis.

The entropy functionals can be evaluated numerically. For instance,  $\mathcal{F}_{\alpha}^{(G)}(x)$ is given on Fig.~\ref{Fig_Expo_classical} for $\alpha>0$.  It is a family of nonnegative functions, equal to zero for $x=m_{Q}=1/\beta$, and  convex for $\alpha\in[1,+\infty)$.

\begin{figure}[ptbh]
\begin{center}
\includegraphics[width=3.6in,height=2.4in]{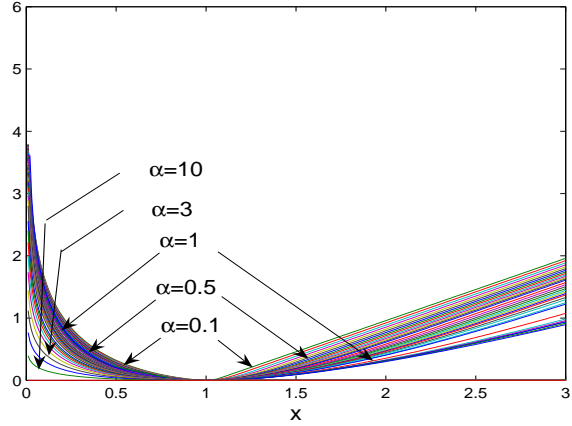}%{fig4.pdf}
\end{center}
\caption{Entropy functional $\mathcal{F}_{\alpha}^{(G)}(x)$ for an exponential
reference measure with $\beta=1$ and $\alpha>0$. By Property \ref{prop:duality} it is also $\mathcal{F}_{1/\alpha}^{(C)}(x)$.}%
\label{Fig_Expo_classical}%
\end{figure}
%%%%%%%%%%%%%%%%%%%%%%%%%%%%%%%%%%%%%%%%%%%%%%%%%%%%%%%%%%%%%%%%
% 
% 
% %%%%%% FIGURE 5 %%%%%%%%%%%%%%%%%%%%%%%%%%%%%%%%%%%%%%%%%%%%%%%%
% \begin{figure}[ptbh]
% \begin{center}
% \includegraphics[width=4in,height=2.8in]{fig5.pdf}
% \end{center}
% \caption{Entropy functional $\mathcal{F}_{\alpha}^{(G)}(x)$ for an
% exponential reference measure with $\beta=1$ and $\alpha\in[0,1]$.}%
% \label{Fig_Expo_generalized}%
% \end{figure}
% %%%%%%%%%%%%%%%%%%%%%%%%%%%%%%%%%%%%%%%%%%%%%%%%%%%%%%%%%%%%%%%%
%
%

\subsection{Bernoulli reference}

Let us now consider the case of the Bernoulli measure $Q(x)=\beta
\delta(x)+(1-\beta)\delta(x-1)$. Of course, the (generalized) mean of optimum distributions is somewhere in the interval $[0,1]$. When $\gamma$ is outside of the interval $(\frac{1}{m-1},\frac{1}{m})$, the probability distribution reduces to a pure state --- $\delta(x)$ or $\delta(x-1)$, and its  (generalized) mean is $0$ or $1$. %Therefore, only the behavior for $\gamma$ in the interval above has interest. 
Incorporation of the bounds into the domain depends on the sign of $\nu:$ for $\nu<0,$ $Z_{\nu}(\gamma,m)$ diverges to $+\infty$ on the bounds whereas it remains finite for $\nu>0.$
The expression of the partition function follows directly from the definition:
% \begin{align}
% Z_{\nu}(\gamma,m) & = \sum_{x=0}^1 \left( 1+\gamma(x-m) \right)^\nu Q(x) \\
% & = \beta(1-\gamma m)^\nu + (1-\beta)(1+\gamma(1-m))^\nu.
% \end{align}
\begin{equation}
 Z_{\nu}(\gamma,m) = \beta(1-\gamma m)^\nu + (1-\beta)(1+\gamma(1-m))^\nu.
\end{equation}

In contrast to the previous case, it is possible here to obtain an explicit expression of the entropy functionals for any $\alpha$. Indeed, if $p$ denotes the value of the optimum distribution at $x=1$, then the generalized expectation is
\begin{align}
m  = \frac{\sum_{x=0}^1 x P(x)^\alpha Q(x)^{1-\alpha}}
{\sum_{x=0}^1 P(x)^\alpha Q(x)^{1-\alpha}} 
 = \frac{(1-\beta)^{1-\alpha}p^\alpha}
{\beta^{1-\alpha}(1-p)^\alpha+(1-\beta)^{1-\alpha}p^\alpha}
\end{align}
and it is therefore possible to express $p$ as a function of $m$:
\begin{equation}
p  = \frac{\left(\beta^{1-\alpha}x\right)^{\frac{1}{\alpha}}}
{\left(\beta^{1-\alpha}x\right)^{\frac{1}{\alpha}} + \left((1-\beta)^{1-\alpha}(1-x)\right)^{\frac{1}{\alpha}}}.
\label{eq:p}
\end{equation}
Now, since the Rényi information divergence is
\begin{equation}
D_\alpha(P||Q) = \frac{1}{\alpha-1} \log
\left[
\beta^{1-\alpha}(1-p)^\alpha+(1-\beta)^{1-\alpha}p^\alpha
\right]
\label{eq:Dbernoulli}
\end{equation}
it suffices to replace $p$ by the expression (\ref{eq:p}) which leads to
\begin{equation}
\mathcal{F}_\alpha^{(G)}(m) =
\frac{\alpha}{1-\alpha} \log
\left[
\beta^{1-\frac{1}{\alpha}}(1-m)^\frac{1}{\alpha} +
(1-\beta)^{1-\frac{1}{\alpha}} m^{1}{\alpha}
\right]
\label{eq:FGbern}
\end{equation}
The case of the classical mean is even simpler: we have $m=p$, and $\mathcal{F}_\alpha^{(C)}(m)$ has the expression of the divergence in (\ref{eq:Dbernoulli}) with $p$ replaced by $m$. It is also interesting to note, and check, 
that the $\alpha\leftrightarrow1/\alpha$ duality of Property \ref{prop:duality} links these two expressions.

The limit case $\alpha\rightarrow1$ is easily derived using L'Hospital's rule. It comes
\begin{equation}
\mathcal{F}_{\alpha\rightarrow1}^{(.)}(x)=x\ln\left(  \frac{x}{1-\beta}\right)  +(1-x)\ln\left(  \frac{1-x}{\beta
}\right).
\end{equation}
This expression is the celebrated\emph{\ Fermi-Dirac entropy} that is strictly convex, nonnegative, and equal to zero for $x=E_{Q}[X]=1-\beta,$ the mean $m_{Q}$ of the reference measure.

Plots of the entropy functionals are given in Figs. \ref{Fig_Bern_classical} and \ref{Fig_Bern_generalized} for $\alpha\in(0,1)$ and $\beta=1/2$. In both cases, we have a family of nonnegative functions, equal to zero for the mean of the reference measure. It can also be checked that $\mathcal{F}_{\alpha}^{(C)}(x)$ is convex for $\alpha\in(0,1]$.

% For the classical mean constraint (C), with $\alpha\in\lbrack0,1],$ $\mathcal{F}%
% _{\alpha}^{(1)}(x)$ is a family of convex functions, minimum for the mean of
% the reference measure $Q$, and restricted to the interval $]0,1[$. \ We can
% also check that for the generalized mean constraint (G), $\mathcal{F}_{\alpha
% }^{(\alpha)}(x)$ is a family of convex functions, also minimum for the mean of
% the reference measure $Q$, and restricted to the interval $[0,1]$.
% 

\begin{figure}
 \begin{minipage}[b]{.46\linewidth}
  \centering\includegraphics[width=3in]{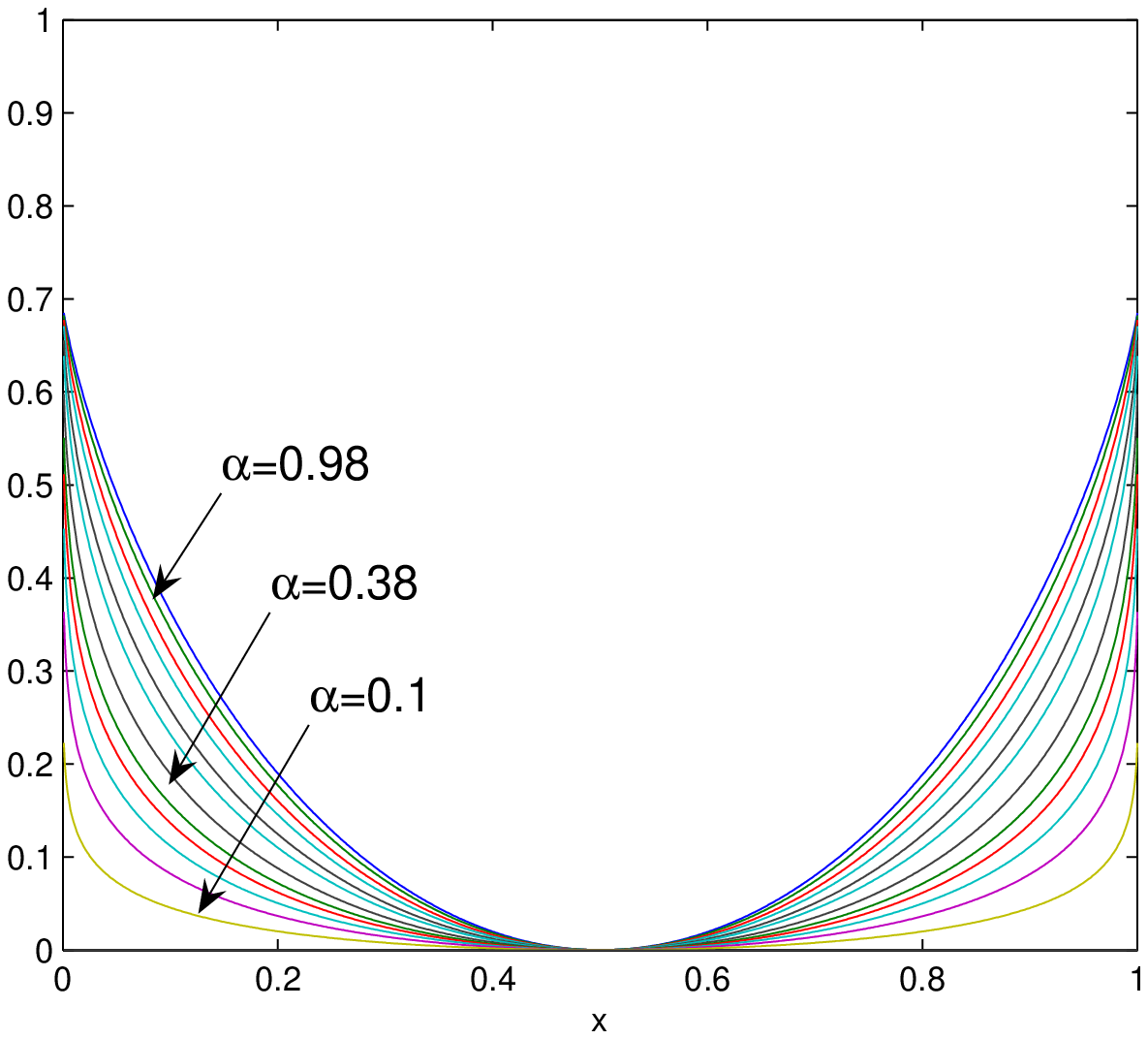}
\caption{Entropy functional $\mathcal{F}_{\alpha}^{(C)}(x)$ \ for a Bernoulli
reference measure and $\alpha\in(0,1)$.}%
\label{Fig_Bern_classical}%
 \end{minipage} \hfill
 \begin{minipage}[b]{.46\linewidth}
  \centering \includegraphics[width=3in]{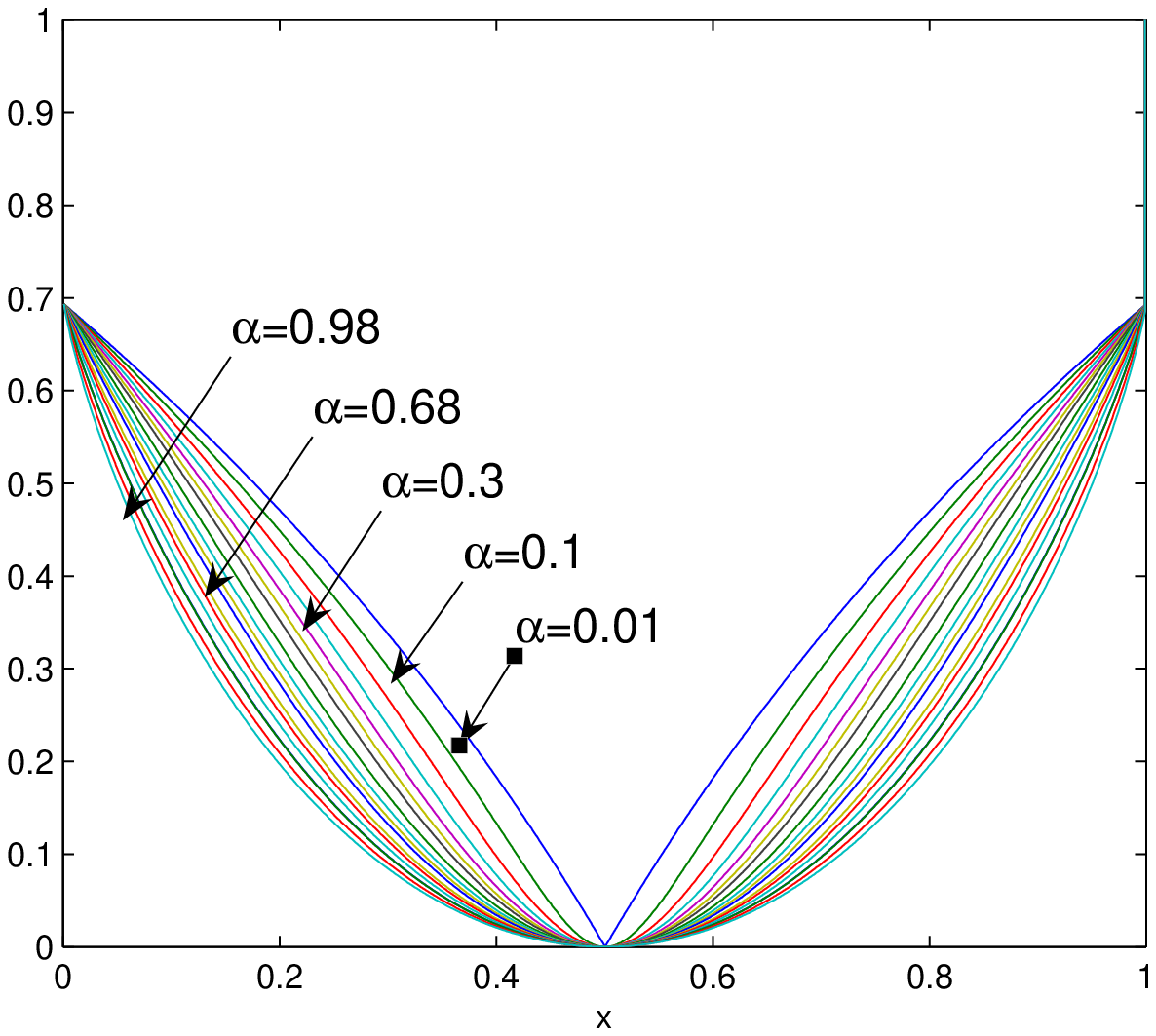}
\caption{Entropy functional $\mathcal{F}_{\alpha}^{(G)}(x)$ \ for a
Bernoulli reference measure and $\alpha\in(0,1)$.}%
\label{Fig_Bern_generalized}%
 \end{minipage}
\end{figure}

% 
% %%%%%% FIGURE 6 %%%%%%%%%%%%%%%%%%%%%%%%%%%%%%%%%%%%%%%%%%%%%%%%%%%%%%%%%%%%%
% \begin{figure}[ptbh]
% \begin{center}
% \includegraphics[width=4in,height=2.8in]{fig6.pdf}
% \end{center}
% \caption{Entropy functional $\mathcal{F}_{\alpha}^{(C)}(x)$ \ for a Bernoulli
% reference measure.}%
% \label{Fig_Bern_classical}%
% \end{figure}
% %%%%%%%%%%%%%%%%%%%%%%%%%%%%%%%%%%%%%%%%%%%%%%%%%%%%%%%%%%%%%%%%%%%%%%%%%%%%%
% 
% 
% %%%%%% FIGURE 7 %%%%%%%%%%%%%%%%%%%%%%%%%%%%%%%%%%%%%%%%%%%%%%%%%%%%%%%%%%%%%
% \begin{figure}[ptbh]
% \begin{center}
% \includegraphics[width=4in,height=2.8in]{fig7.pdf}
% \end{center}
% \caption{Entropy functional $\mathcal{F}_{\alpha}^{(G)}(x)$ \ for a
% Bernoulli reference measure.}%
% \label{Fig_Bern_generalized}%
% \end{figure}
% %%%%%%%%%%%%%%%%%%%%%%%%%%%%%%%%%%%%%%%%%%%%%%%%%%%%%%%%%%%%%%%%%%%%%%%%%%%%%

\subsection{Poisson reference}

As a final example, let us consider the case of a Poisson measure 
$Q(x)=\frac{\mu^{x}}{x!}e^{-\mu},$ for $x\geq0$. 
Domain $\mathcal{D}$ is $\mathcal{D=D}% 
_{Q}\cap\mathcal{D}_{\gamma},$ where $\mathcal{D}_{Q}=%
\mathbb{N}^{+}$ and $\mathcal{D}_{\gamma}=\left\{  x:\gamma(x-m)+1\geq0\right\}$. The
partition function is given by
\begin{equation}
Z_{\nu}(\gamma,m)=\sum_{\mathcal{D}}\left[  \gamma(x-m)+1\right]^{\nu}%
\frac{\mu^{x}}{x!}e^{-\mu}. 
\label{eq:partition_function_poisson}%
\end{equation}
Three cases appear, according to the value of $\gamma$:
\begin{itemize}
%First, 
\item[(a)] if $\frac{1}{m} \ge \gamma \ge 0$, then   
$\mathcal{D}$ reduces to $\mathcal{D}_{1}=\left\{  x:x\in\left[  0,+\infty\right) \right\}$; 
%second, 
\item[(b)] for 
$\gamma \ge \frac{1}{m}$ the domain is $\mathcal{D}_{2}   =\left\{  x:x\in [  \left\lceil m-\frac{1}{\gamma
}\right\rceil ,+\infty ) \right\}$; %and 
%third 
\item[(c)] when $\gamma<0$, $\mathcal{D}=\mathcal{D}_{3} 
=\left\{  \text{ }x\in\lbrack 0,\left\lfloor m-\frac {1}{\gamma}\right\rfloor ]
\right\}$. 
\end{itemize}
In these expressions 
 $\lfloor x\rfloor$ denotes the floor function that
returns the largest integer less than or equal to x; and $\lceil x\rceil$
is the ceil function, the smallest integer not less than $x$.
% $\left\lceil x\right\rceil $ gives the smallest integer $\geq x$ (ceil
% function) and $\left\lfloor x\right\rfloor $ denotes the greatest integer $<x$ (the integer part, floor function).

Closed-form formulas can not be derived in the general case, but only in the
case of an integer exponent $\nu.$ 
% In such a case, the partition function can be expressed using generalized hypergeometric functions, as shown below. 
When $\nu$ is not an integer, we will have to resort to the serie
(\ref{eq:partition_function_poisson}), possibly truncated for numerical
computations. In order to save space, we only sketch the derivation in
$\mathcal{D}_{1}$: 
%
%In the first domain the partition function can be rewritten as 
\begin{equation}
Z_{\nu}(\gamma,m)=\left(  1-\gamma m\right)  ^{\nu}e^{-\mu}\sum_{x=0}%
^{+\infty}\left(  \theta x+1\right)  ^{\nu}\frac{\mu^{x}}{x!}%
%%\ \ \ \ \ \text{with }\theta=\frac{\gamma}{1-\gamma m}.
\label{eq:partition_function_poisson_D1}%
\end{equation}
with $\theta=\frac{\gamma}{1-\gamma m}$. 
In the serie above 
%$\sum_{x=0}^{+\infty}\left(  \theta x+1\right)  ^{\nu }\frac{\mu^{x}}{x!}$ 
the ratio of successive terms $\frac{\left(  1+\theta
x+\theta)\right)  ^{\nu}}{(x+1)\left(  1+\theta x\right)  ^{\nu}}\mu$ is the
ratio of two completely factored polynomials. This indicates that the serie
can be written as a generalized hypergeometric function, when $\nu$ is
integer. So doing, we obtain%
\begin{align*}
Z_{\nu}(\gamma,m) &  =\left(  1-\gamma m\right)  ^{\nu}e^{-\mu
}\ _{|\nu|}F_{|\nu|}(a,...,a_{;}b,...,b_{;}\mu)
% \\
% \text{with }a  &  =\frac{1+\theta}{\theta}\text{\ and }b=\frac{1}{\theta
% }\text{ \ for \ }\nu>0\\
% \text{with }a  &  =\frac{1}{\theta}\text{ and }b=\frac{1+\theta}{\theta}\text{
% for }\nu<0.
\end{align*}
with $a={(1+\theta)}/{\theta}$ and $b=1/\theta$ for $\nu>0$; or with
$a=1/\theta$ and $b={(1+\theta)}/{\theta}$ for $\nu<0$.

The derivative with respect to $\gamma$ is%
\begin{equation}
%\begin{array}{lll}
\frac{d}{d\gamma}Z_{\nu+1}( \gamma  ,m) = \left(  1-\gamma m\right)
^{\nu}e^{-\mu}\left(  \nu+1\right)  
%\\ & \times & 
 \sum_{x=0}^{+\infty}(x-m)\left(  1+\theta
x\right)  ^{\nu}\frac{\mu^{x}}{x!},
%\end{array}
\label{eq:derivative_partition_function_poisson_D1}%
%\right.
\end{equation}
that can also be expressed using hypergeometric functions.
%
%
%\begin{align*}
%\frac{d}{d\gamma}Z_{\nu+1}(\gamma,m)  &  =\left(  1-\gamma m\right)  ^{\nu
%}e^{-\mu}\left(  \nu+1\right)  \left(
%\begin{array}
%[c]{c}%
%\left(  1+\theta\right)  ^{\nu}\mu_{\nu}F_{\nu}\left(  \left[  a_{1}%
%,...,a_{1}\right]  ,\left[  b_{1},...,b_{1}\right]  ,\mu\right)
%\allowbreak-\\
%m_{\nu}F_{\nu}\left(  \left[  a_{2},...,a_{2}\right]  ,\left[  b_{2}%
%,...,b_{2}\right]  ,\mu\right)
%\end{array}
%\right)  \allowbreak\\
%\text{with \ }a_{1}  &  =\frac{2\theta+1}{\theta}\ b_{1}=\frac{1+\theta
%}{\theta}\ a_{2}=\frac{1+\theta}{\theta}\ b_{2}=\frac{1}{\theta}\text{\ for
%\ }\nu>0,\text{ \ and}\\
%\text{with \ }a_{1}  &  =\frac{1+\theta}{\theta}\ \ b_{1}=\frac{2\theta
%+1}{\theta}\ a_{2}=\frac{1}{\theta}\ b_{2}=\frac{1+\theta}{\theta}\ \text{for
%\ }\nu<0.
%\end{align*}
Formulas for domains $\mathcal{D}_{2}$ and $\mathcal{D}_{3}$ also involve
hypergeometric functions. 
With these formulas, or by direct
evaluation of (\ref{eq:partition_function_poisson}), functionals $D_C(\gamma)$ and $D_G(\gamma)$ can be
evaluated and maximized on their domains of definition so as to find the optimum value $\gamma^*$. 

Given the signs of $\nu$ and $\gamma$, and the supports $\mathcal{D}_{1}$, $\mathcal{D}_{2}$ and $\mathcal{D}_{3}$, 
it is already possible to deduce that the solution $\gamma^*$ is necessarily in a specific interval. Hence, we obtain here that for $\nu>0$ (respectively for $\nu<0$), solutions associated to a constraint  $m>\mu$ corresponds to case (a) (resp. case (c)) and that solutions for $m<\mu$ correspond to case (c) (resp. case (a)). The argumentation relies on the fact that if $P_i$ and $P_j$ are two optimum distributions with supports $\mathcal{D}_{i}$ and $\mathcal{D}_{j}$, with the same (generalized) mean  but different parameters, then by Theorem \ref{theo_solgen} $D_\alpha(P_j||Q)\ge D_\alpha(P_i||Q)$ if $P_j$ is dominated by $P_i$.

In the case $m>\mu$, $\nu<0$, 
the solution with minimum divergence is for a distribution $P_3$ in case (c), and furthermore we have $D_{\alpha}(P_{3\text{ }}||Q)\rightarrow0$. 
This can be seen as follows. Let $x\in\mathcal{D}_{3}$ and $k=\left\lfloor m-\frac{1}{\gamma}\right\rfloor
,$ so that $x<k+1.$ Let now $\gamma=\frac{1}{m-k}+\epsilon$ with $\epsilon
\in\left(  0,\frac{1}{m-k-1}-\frac{1}{m-k}\right).$ 
Then the mean of the distribution is given by 
\begin{equation}
E_{\nu}\left[  X\right]  =\frac{1}{Z_{\nu}(\gamma,m)}\sum_{x=0}^{k-1}x\left[
\frac{k-x}{k-m}\right]  ^{\nu}Q(x)+k\frac{\left[  k-m\right]  ^{\nu}}{Z_{\nu
}(\gamma,m)}\epsilon^{\nu}Q(k),
\end{equation}
and any value higher than $\mu=E_{Q}[X]$ can be obtained by tuning $\epsilon,$ for many
values of $k$.
%, when $Q(k)$ $<$ $k^{-\nu}$ (short tails). 
When $k$ increases$,\gamma=\frac{1}{m-k}$ tends to $0$ by lower values and $P_3$ tends to $Q$, which results in $D_{\alpha }(P_{3\text{ }}||Q)\rightarrow0.$

The $\nu<0$ case has the specificity that $Z_{\nu}(\gamma,m)$ exhibits singularities at
$\gamma=\frac{1}{m-k}$ for all $k\geq0$. Then $Z_{\nu}(\gamma,m)$, with $\nu=-(\xi+1)$ or $\nu=\xi$, 
is only convex on intervals $[\frac{1}{m-k},\frac{1}{m-k-1}]$ or $[\frac{1}{m-k-1},\frac
{1}{m-k}]$ (for $k+1>m>k$), with $Z_{\nu}(\gamma,m)=+\infty$ on the bounds of
each interval. Consequently, $-\log Z_{\nu}(\gamma,m)$ may present several
maxima. This is illustrated in Fig. \ref{Fig_Pois_dual_exampl} where 
function $D_C(\gamma)$ with $\alpha=0.5$
%, computed for $\mu=3, \alpha=0.5$ and $m=1.15$ is concave on
% subdomains and 
presents many extrema. The solution with minimum Rényi divergence 
corresponds to the minimum of these maxima. 
%, with two maxima at $\gamma=0.35$ and
%$\gamma=1.24$.

% %%%%%% FIGURE 8 %%%%%%%%%%%%%%%%%%%%%%%%%%%%%%%%%%%%%%%%%%%%%%%%%%%%%%%%%%
% \begin{figure}[ptbh]
% \begin{center}
% \includegraphics[width=3.2in]{fig8.pdf}
% \end{center}
% \caption{Example of functional $D_C(\gamma)$ for the Poisson reference with classical
% mean constraint, with $\mu=3, \alpha=0.5$ ($\xi=-2$) and $m=1.15$. It presents 
% singularities at $1/(m-k)$, $\forall k$, and maxima 
% at $\gamma=0.35$ and $\gamma=1.24$.}%
% \label{Fig_Pois_dual_exampl}%
% \end{figure}
% %%%%%%%%%%%%%%%%%%%%%%%%%%%%%%%%%%%%%%%%%%%%%%%%%%%%%%%%%%%%%%%%%%%%%%%%%%

The limit case $\alpha\rightarrow1$ is obtained with $\left\vert
\nu\right\vert =\left\vert \xi\right\vert $ $\rightarrow+\infty$. \ 
According to the discussion above, the optimum $\gamma$ corresponds to case (a)
 for $\{m>\mu,\nu>0\}$ and $\{m<\mu,$ $\nu<0\},$ and to case (c)
for $\{m<\mu,\nu>0\}.$ For case (a), the support is $\mathcal{D}_{1},$ and the derivative of the partition
function $Z_{\nu}(\gamma,m)$ is given by (\ref{eq:derivative_partition_function_poisson_D1}). In this derivative, the sum can be rewritten as
\begin{equation}
\begin{array}{l}
\displaystyle{\sum_{x=0}^{+\infty}\left(  x-m\right)  \left(  1+\theta x\right)  ^{\nu}%
\frac{\mu^{x}}{x!}} 
%\\
\displaystyle{=\sum_{x=0}^{+\infty}\left(  \mu\left(  1+\theta
x+\theta\right)  ^{\nu}-m\left(  1+\theta x\right)  ^{\nu}\right)  \frac
{\mu^{x}}{x!}},%
\end{array}
\label{eqequal}
\end{equation}
so that $Z_{\nu+1}(\gamma,m)$ is minimum when the RHS of (\ref{eqequal}) is equal to zero.
%
%\begin{equation}
%\sum_{x=0}^{+\infty}\left(  \mu\left(  1+\theta x+\theta\right)  ^{\nu
%}-m\left(  1+\theta x\right)  ^{\nu}\right)  \frac{\mu^{x}}{x!}=0.
%\end{equation}
We have to solve this equation in $\theta.$ Suppose that $\theta$ is small and
that $\theta x\ll1$ for the significative values of the probability
distribution.
% When $x$ increases, the probability distribution decreases to
%zero and the contributions for $\theta x\not \ll 1$ will be negligible.  
In
this case, we use the approximation $\left(  1+\theta x\right)  ^{\nu}%
=e^{\nu\log\left(  1+\theta x\right)  }\approx e^{\nu\theta x},$  that
leads to 
\begin{equation}
\sum_{x=0}^{+\infty}\left(  \mu e^{\nu\theta(x+1)}-me^{\nu\theta x}\right)
\frac{\mu^{x}}{x!}=e^{\mu e^{\nu\theta}}\left(  \mu e^{\nu\theta}-m\right)=0
\end{equation}
%Then, observing that $\sum_{x=0}^{+\infty}e^{\nu\theta x}\frac{\mu^{x}}%
%{x!}=\sum_{x=0}^{+\infty}\left(  \mu e^{\nu\theta}\right)  ^{x}\frac{1}%
%{x!}=e^{\mu e^{\nu\theta}},$ the previous equation reduces to%
%%
%\begin{equation}
%e^{\mu e^{\nu\theta}}\left(  \mu e^{\nu\theta}-m\right)  =0.
%\end{equation}
The solution is given by $\theta^{\ast}=\frac{1}{\nu}\log(\frac{m}{\mu}),$
that in turn provides
\begin{equation}
\gamma^{\ast}=\frac{\ln\frac{m}{\mu}}{\nu+m\ln\frac{m}{\mu}}.
\label{eq:gamma_opt_poisson}%
\end{equation}
In case (a), $\gamma$ is positive, and this will be true for
$\gamma^{\ast}$ if $\{m>\mu,\nu>0\}$ or $\{m<\mu,\nu<0\}$. % as observed in  \ref{seq:behavior_in_the_three_domains}. %Also note
%that $\theta^{\ast}$ and $\gamma^{\ast}$ are indeed small for $\left\vert
%\nu\right\vert \rightarrow+\infty$ \ so that the approximation $\theta x\ll1$
%is valid. 
For the log-partition function, when $|\nu|\rightarrow+\infty$, this leads to%
\begin{equation}
-\log\ Z_{\nu+1}(\gamma^{\ast},m)
% =(\nu+1)\log\left(  1+\frac{m}{\nu}\log
%\frac{m}{\mu}\right)  -\mu\left(  \frac{m}{\mu}\right)  ^{\frac{\nu+1}{\nu}%
%}+\mu
\approx m\log\frac{m}{\mu}+(\mu-m).
\end{equation}
%with $\frac{1}{\nu}\log(\frac{m}{\mu})x\ll1$ and $\frac{\nu+1}{\nu}\approx1$.

In domain $\mathcal{D}_{3},$ the derivative of the partition function $Z_{\nu
}(\gamma,m)$ is equal to zero if
\[
\sum_{x=0}^{k}\left(  x-m\right)  \left(  1+\theta x\right)  ^{\nu}\frac
{\mu^{x}}{x!}=0,\text{ with  \ } k=\left\lfloor m-\frac{1}{\gamma
}\right\rfloor, \gamma<0.
\]
If $\gamma$ is small enough, $k\rightarrow+\infty$ and we obtain for $\nu>0$
the same formulation and solution as in $\mathcal{D}_{1}$ 
%(for $\nu<0,$ the approximation $\theta x\ll1$ is not valid). 
The solution $\gamma^{\ast}$ in
(\ref{eq:gamma_opt_poisson}) is now negative, that imposes $m<\mu$ for
$\nu>0.$ Finally, we have shown above that if $m>\mu$ with $\nu<0$ then
$D_{\alpha}(P_{3\text{ }}||Q)\rightarrow0.$

Hence, we obtain that the entropy functionals converge to
% \begin{equation}
% \left\{
% \begin{array}
% [c]{lcl}%
% \mathcal{F}_{\alpha\rightarrow1^{-}}^{(C)}(x) &=&\left\{
% \begin{array}
% [c]{cc}%
% x\ln\frac{x}{\mu}+(\mu-x) & \text{for }x<\mu\\
% 0 & \text{for }x>\mu
% \end{array}
% \ \right.  \\
% \mathcal{F}_{\alpha\rightarrow1^{+}}^{(C)}(x) &=&x\ln\frac{x}{\mu}+(\mu-x)
% \end{array}
% \right.
% \end{equation}
\begin{equation}
 \mathcal{F}_{\alpha\rightarrow1}^{(.)}(x) = x\ln\frac{x}{\mu}+(\mu-x)
% %  \mathcal{F}_{\alpha\rightarrow1}^{(.)}(x) = \left\{
% %  \begin{array}
% %  [c]{cc}%
% %  x\ln\frac{x}{\mu}+(\mu-x) & \\
% %  0 & \text{for }x>\mu \text{if (C) $\alpha<1$ or (G) $\alpha>1$}
% %  \end{array}
% % \right.
\end{equation}
with the restriction that $\mathcal{F}_{\alpha}^{(.)}(x)=0$ for $x>\mu$ if (C) $\alpha<1$ or (G) $\alpha>1$. 

This functional is simply the cross-entropy between $x$ and $\mu$ or
Kullback-Leibler (Shannon) entropy functional with respect to\emph{\ }$\mu$
\cite{Csiszar1991}. \ It measures a \ `distance' between a
possible mean (observable) and a reference mean $\mu$, and it has been used as a
regularization functional in several applied problems, such as astronomy,
tomography, RMN, and spectrometry.

As in the previous cases, the entropy functionals $\mathcal{F}_{\alpha}^{(C)}(x)$ and $\mathcal{F}%
_{\alpha}^{(G)}(x)$ can be evaluated numerically.  For instance, $\mathcal{F}%
_{\alpha}^{(G)}(x)$ is given on  Fig. \ref{Fig_Pois_generalized}
for $\mu=3$. It presents an unique minimum for $m=\mu,$ and we note that it is is not convex for small values of $\alpha.$

%%%%%%% FIGURE 9 %%%%%%%%%%%%%%%%%%%%%%%%%%%%%%%%%%%%%%%%%%%%%%%%%%%%%%%%%
%\begin{figure}[ptbh]
%\begin{center}
%\includegraphics[width=4in,height=2.8in]{fig9.pdf}
%\end{center}
%\caption{Entropy functional $\mathcal{F}_{\alpha}^{(C)}(x)$ for a Poisson
%reference measure with $\mu=3$ and $\alpha\in[0,1]$.}%
%\label{Fig_Poisson_classical}%
%\end{figure}
%%%%%%%%%%%%%%%%%%%%%%%%%%%%%%%%%%%%%%%%%%%%%%%%%%%%%%%%%%%%%%%%%%%%%%%%%

%%%%%% FIGURE 10 %%%%%%%%%%%%%%%%%%%%%%%%%%%%%%%%%%%%%%%%%%%%%%%%%%%%%%%

\begin{figure}
 \begin{minipage}[b]{.46\linewidth}
 \begin{center}
\includegraphics[width=3in, height=2.1in]{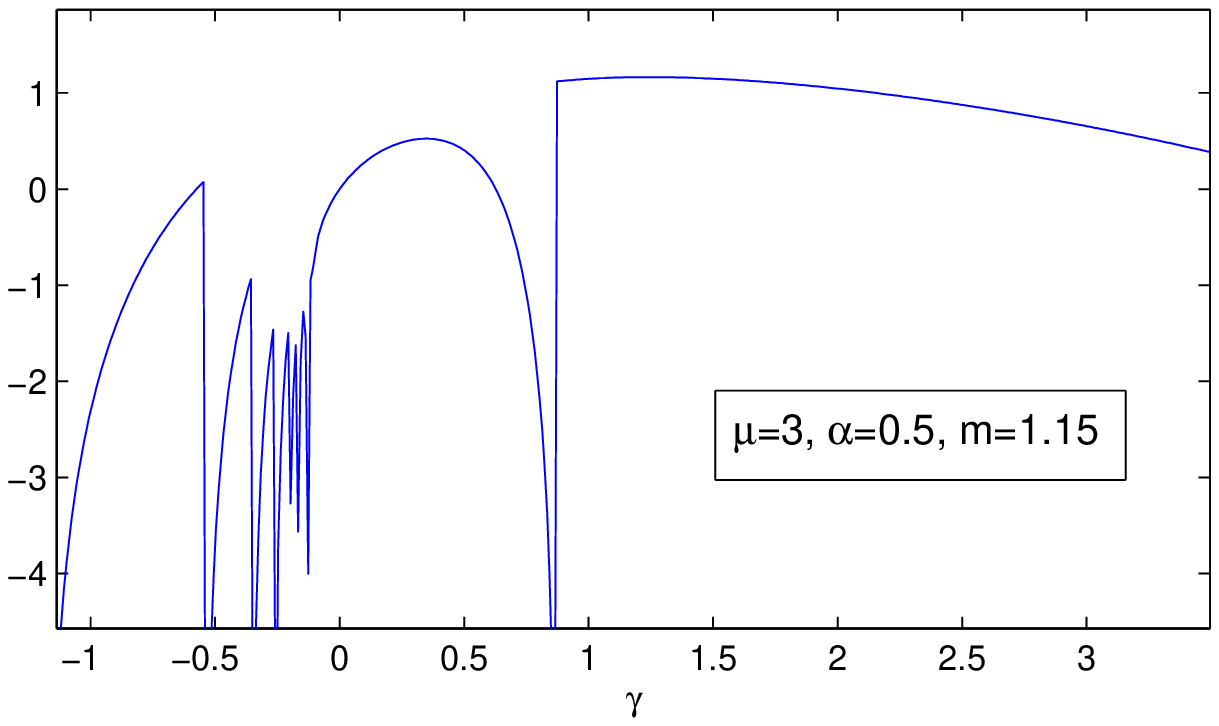}
\end{center}
\caption{Example of functional $D_C(\gamma)$ for the Poisson reference with classical
mean constraint, with $\mu=3, \alpha=0.5$ ($\xi=-2$) and $m=1.15$. It presents 
singularities at $1/(m-k)$, $\forall k$, and maxima 
at $\gamma=0.35$ and $\gamma=1.24$.}%
\label{Fig_Pois_dual_exampl}%
 \end{minipage} \hfill
 \begin{minipage}[b]{.46\linewidth}
  \begin{center}
\includegraphics[width=3in, height=2.1in]{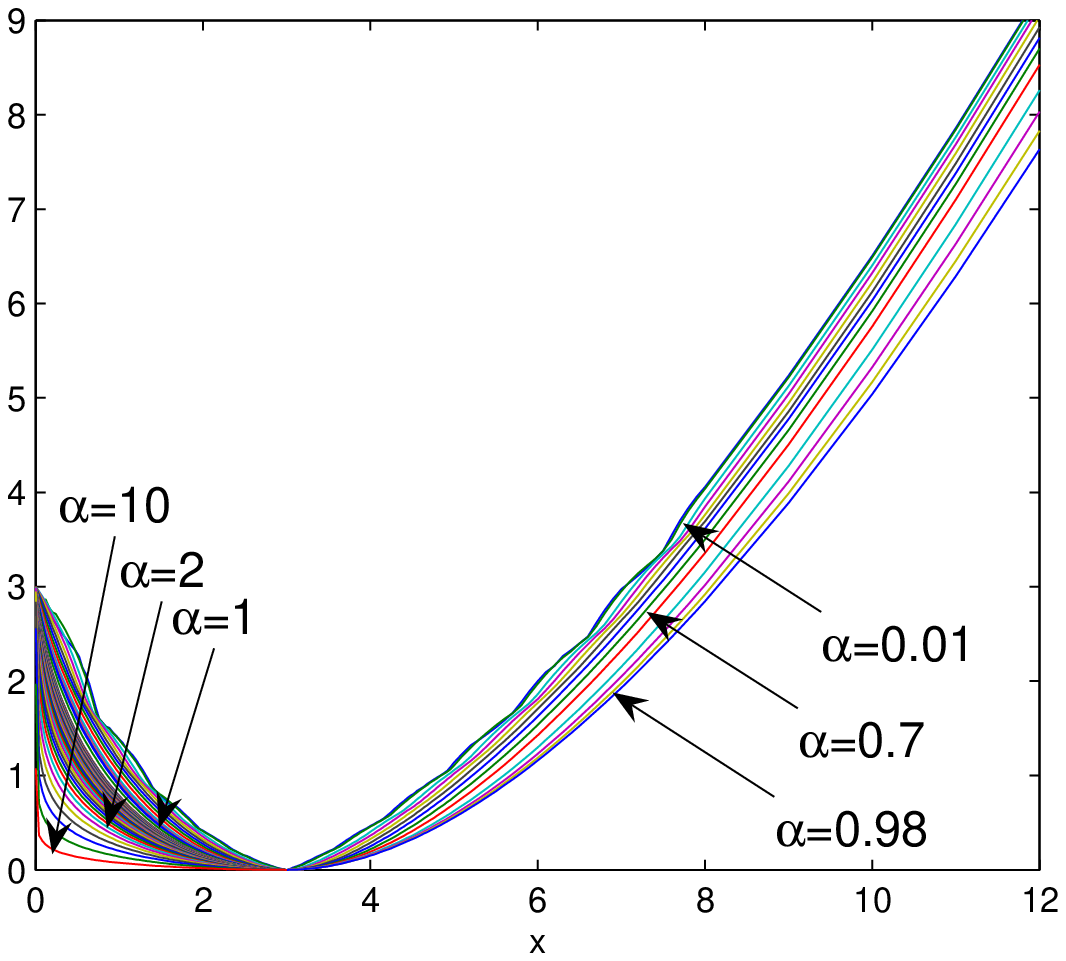} %%fig10.pdf
\end{center}
\caption{Entropy functional $\mathcal{F}_{\alpha}^{(G)}(x)$ for a Poisson
reference measure with $\mu=3$ and $\alpha\ge0$. For $\alpha\rightarrow 1$, $\mathcal{F}_{\alpha}^{(.)}(x)$ 
converges to $x\ln\frac{x}{\mu}+(\mu-x)$.}%
\label{Fig_Pois_generalized}%
 \end{minipage}
\end{figure}

% % %%%%%% FIGURE 8 %%%%%%%%%%%%%%%%%%%%%%%%%%%%%%%%%%%%%%%%%%%%%%%%%%%%%%%%%%
% % \begin{figure}[ptbh]
% % \begin{center}
% % \includegraphics[width=3.2in]{fig8.pdf}
% % \end{center}
% % \caption{Example of functional $D_C(\gamma)$ for the Poisson reference with classical
% % mean constraint, with $\mu=3, \alpha=0.5$ ($\xi=-2$) and $m=1.15$. It presents 
% % singularities at $1/(m-k)$, $\forall k$, and maxima 
% % at $\gamma=0.35$ and $\gamma=1.24$.}%
% % \label{Fig_Pois_dual_exampl}%
% % \end{figure}
% % %%%%%%%%%%%%%%%%%%%%%%%%%%%%%%%%%%%%%%%%%%%%%%%%%%%%%%%%%%%%%%%%%%%%%%%%%%
% % 
% % 
% % 
% % \begin{figure}[ptbh]
% % \begin{center}
% % \includegraphics[width=3.4in,height=2.2in]{fig10.pdf}
% % \end{center}
% % \caption{Entropy functional $\mathcal{F}_{\alpha}^{(G)}(x)$ for a Poisson
% % reference measure with $\mu=3$ and $\alpha\in[0,1]$.}%
% % \label{Fig_Pois_generalized}%
% % \end{figure}
% % 
% % 
% % %%%%%%%%%%%%%%%%%%%%%%%%%%%%%%%%%%%%%%%%%%%%%%%%%%%%%%%%%%%%%%%%%%%%%%%

%%%%%%%%%%%%%%%%%%%%%%%%%%%%%%%%%%%%%%%%%%%%%%%%%%%%%%%%%%%%%%%
% % \subsection{Little things}
% % 
% % \begin{lemma}
% % \label{l7_0} As a direct consequence of the previous lemma, we obtain that
% % $Z_{\nu+1}(\gamma,\overline{x})=Z_{\nu}(\gamma,\overline{x})$ if and only if
% % $\overline{x}=E_{\nu}\left[  x\right]  .$ When $\overline{x}$ is a fixed
% % parameter $m,$ this will be only true for a special value $\gamma^{\ast}$ such
% % that $E_{\nu}\left[  x\right]  =m.$
% % \end{lemma}

%%% --> Conclusion

\section{Conclusion and future work}

By weakening one of the postulates that lead to the definition of Shannon entropy, Rényi \cite{Renyi1961} introduced a one parameter family of entropy and divergence. Shannon entropy and Kullback-Leibler divergence are recovered in the limiting case for the parameter $\alpha\rightarrow1$. In this work, we considered the maximum entropy problems associated with Rényi $Q$-entropies. We characterized the solutions for a standard mean constraint and for the generalized mean constraint of nonextensive statistics. We defined and discussed the entropy functionals as a function of the constraints. These entropies were characterized and various properties and relationships were highlighted. We also discussed numerical aspects. Finally we illustrated this setting through some specific examples and recovered some well-kown entropy functionals. 

Future work will consider the extension of this setting in the multivariate case. An issue that should be examined is the fact that the direct multivariate extension of (\ref{eq:defPnu}) is not separable in the case of a separable reference $Q(x)$; which means that some dependances are implicitely introduced in the maximum entropy solution.

We also intend to investigate a possible underlying geometrical structure of the maximum entropy distributions (\ref{eq:defPnu}). This structure should extend the geometrical structure of exponential families and involve the Bregman-like divergence introduced by (\ref{eq:bregman-like}). 

Finally, maximum entropy methods have been successfully employed for solving inverse problems. We intend to consider the potential of Rényi entropies and divergence in this field. A simple contribution would be to examine the interest of a Rényi entropy functional, e.g. (\ref{eq:FGbern}),  as a potential in a Markov field for image deconvolution or restoration.

%%% <-- Conclusion
%%%%%%%%%%%%%%%%%% Appendix A
\appendix

\section{Proof of Theorem \ref{theo_solgen}}
\label{proof:theo} 
Let us begin with the classical constraint (C). In this first case, we follow the approach of \cite{Vignat2004}. Consider the functional Bregman divergence :
\[
 B_h(f,g)=\int d(f,g) h(x) dx = \int -\left( f(x)^{\alpha} - g(x)^{\alpha} - {\alpha}\left(f(x)-g(x)\right)g(x)^{\alpha-1} \right) h(x) dx
\]
 where $h(x)$ is a nonnegative functional, associated to the (pointwise) Bregman divergence $d(f,g)$ built upon the strictly convex function $-x^{\alpha}$ for $\alpha\in(0,1)$. Then
% \[
%  d(f,g)=f(x)^{\alpha} - g(x)^{\alpha} - {\alpha}\left(f(x)-g(x)\right)g(x)^{\alpha-1}
% \]
\begin{align}
B_{Q^{1-\alpha}}(P,P_C)  & = -\int_\mathcal{S} P(x)^{\alpha} - P_C(x)^{\alpha} - {\alpha}(P(x) P_C(x)^{\alpha-1}-P_C(x)^{\alpha}) Q(x)^{1-\alpha} dx\\
%& = -\int P(x)^{\alpha} Q^{1-\alpha} dx+ \int {P_C}(x)^{\alpha}Q^{1-\alpha} dx + \alpha \int {P_C}(x)^{\alpha}Q^{1-\alpha} dx - \alpha\int P_C(x)^{\alpha}Q^{1-\alpha} dxdx,\\
& = -\int_\mathcal{S} P(x)^{\alpha}Q(x)^{1-\alpha} dx + \int_\mathcal{S} {P_C}(x)^{\alpha}Q(x)^{1-\alpha} dx. \label{eq:Dpos}
\end{align}
with $h(x)=Q(x)^{1-\alpha}$ and where $\mathcal{S}$ denotes the support of $P_C(x)$. The second line follows from the fact that when $P$ and $P_C$ have the same mean $\bar{x}=E_{P_C}[X]=E_{P}[X]$, then using the expression in (\ref{eq:defPnu}) with $\nu=\xi=\frac{1}{\alpha-1}$ it is possible to check that
\[
\int_\mathcal{S} P(x) P_C(x)^{\alpha-1} Q(x)^{1-\alpha}dx = \int_\mathcal{S} P_C(x)^{\alpha} Q(x)^{1-\alpha}dx = Z_\xi(\gamma,\bar{x})^{-\alpha}
\]
 %
%
% \begin{align}
%  B(f,g) & =-\int f(x)^{\alpha} - g(x)^{\alpha} - {\alpha}\left(f(x)-g(x)\right)g(x)^{\alpha-1} dx,\\
% & = -\int f^{\alpha} - g^{\alpha} - {\alpha}(fg^{\alpha-1}-g^{\alpha}).
% \end{align}
% Because of the expression of $P_C(x)$, it is easy to check that when $P$ and $P_C$ have the same mean, that is $\bar{x}=E_{P_C}[X]=E_{P}[X]$, then we have
% $\int P(x) P_C(x)^{\alpha-1} dx = \int P_C(x)^{\alpha} dx$. In such case,
%
% \begin{align*}
%  B(P,P_C)  & = -\int P(x)^{\alpha} - P_C(x)^{\alpha} - {\alpha}(P(x) P_C(x)^{\alpha-1}-P_C(x)^{\alpha}) dx\\
% & = -\int P(x)^{\alpha} + \int {P_C}(x)^{\alpha} + \alpha \int {P_C}(x)^{\alpha} - \alpha\int P_C(x)^{\alpha}dx,\\
% & = -\int P(x)^{\alpha}dx + \int {P_C}(x)^{\alpha}dx.
% \end{align*}
provided the whole support of $P(x)$ is included in $\mathcal{S}$, which is the case by the absolute continuity of $P(x)$ with respect to $P_C(x)$.

The Bregman divergence $B_{Q^{1-\alpha}}(P,P_C)$ being always positive and equal to zero if and only if $P=P_C$, the equality (\ref{eq:Dpos}) implies that, for $\alpha\in(0,1)$,
\begin{equation}
  D_\alpha(P||Q) \geq D_\alpha(P_C||Q)
\end{equation}
which means that $P_C$ is the distribution with minimum Rényi (Tsallis) divergence to $Q$, in the set of all distributions $P \ll P_C$ with a given mean $\bar{x}$, for $\alpha\in(0,1)$. The case $\alpha>1$ can be derived accordingly, beginning with the Bregman divergence associated to the strictly convex function $x^\alpha$.\\

As far as the generalized mean constraint (G) is concerned, let us now consider the Rényi information divergence $D_\alpha(P||P_G)$ from $P$ to $P_G$, with $P_G$ given in (\ref{eq:defPnu}) with $\nu=-\xi=\frac{1}{1-\alpha}$
\begin{equation}
 (\alpha-1)   D_\alpha(P||P_G)  = \log \int_\mathcal{S} P(x)^\alpha P_G(x)^{1-\alpha} dx, \\
\end{equation}
with $\mathcal{S}$ the support of $P_G(x)$, and which can be rearranged as
\begin{align}
%(\alpha-1) &  D_\alpha(P||P_G)  = \log \int P(x)^\alpha P_G(x)^{1-\alpha} dx \\
(\alpha-1)  D_\alpha(P||P_G)  & = \log \int_\mathcal{S} \frac{P(x)^\alpha{Q(x)^{1-\alpha}}}{\int_\mathcal{S} P(x)^\alpha{Q(x)^{1-\alpha}} dx } \left[  \gamma(x-\overline{x})+1\right] dx \\
& + \log \int_\mathcal{S} P(x)^\alpha{Q(x)^{1-\alpha}}dx - (1-\alpha) \log Z_{\frac{1}{1-\alpha}}(\gamma,\overline{x}).
\end{align}
The generalized mean with respect to $P$ appears in the first term, and cancels if $P$ and $P_G$ have the same generalized mean $\bar{x}$ and $P_G \gg P$. In such a case, we obtain
\begin{align}
D_\alpha(P||P_G) & = \frac{1}{(\alpha-1)}
\log \int P(x)^\alpha Q^{1-\alpha}dx  +  \log Z_{\frac{1}{1-\alpha}}(\gamma,\overline{x})\\
& = D_\alpha(P||Q) - D_\alpha(P_G||Q), \label{eq:pythagorean}
\end{align}
where we used the fact that $D_\alpha(P_G||Q)=-\log Z_{\frac{1}{1-\alpha}}(\gamma,\overline{x})$ as stated in Proposition \ref{prop:val_opt_renyi}. Since the Rényi information divergence is always greater or equal to zero, we have
\begin{gather}
 D_\alpha(P||Q) \geq D_\alpha(P_G||Q)
\end{gather}
and conclude that $P_G$ is the distribution with minimum Rényi (Tsallis) divergence to $Q$, in the set of all distributions $P \ll P_G$ with a given generalized $\alpha$-mean $\bar{x}$.
%  since
% \begin{gather}
%  D_\alpha(P||Q) \geq D_\alpha(P_G||Q)
% \end{gather}
% for any distribution $P$ with same generalized $\alpha$-mean as $P_G$.

Finally, it is easy to check, given the expression of $P_G$ and the fact that $\alpha\xi=\xi+1$, that the generalized mean of $P_G$ is also the standard mean of the distribution with exponent $\nu=-(\xi+1)$, that is $E_{P_G}^{(\alpha)}\left[X\right]=E_{P_G^{\ast}}[X]=E_{-(\xi+1)}[X]$.
\\
Note that the equality in (\ref{eq:pythagorean}), $ D_\alpha(P||Q) = D_\alpha(P||P_G) + D_\alpha(P_G||Q)$, is a pythagorean equality, which means that $P_G$ is the orthogonal projection of $P$ on the set of probability distributions with fixed generalized mean $\bar{x}$.\\

\section{Proof of Proposition \ref{prop:derivate}}
\label{proof:propderivate}

The exact behaviour depends on the reference distribution $Q(x)$ and on the sign of the exponent $\nu$. Because the domain of definition $\mathcal{D}$ might depend on $\gamma$, the derivative of the partition function writes
\[
\frac{dZ_{\nu}(\gamma,\bar{x}_{\gamma})}{d\gamma}=\lim_{\delta\gamma\rightarrow0}\frac{1}{\delta\gamma}\left(Z_{\nu}(\gamma+\delta\gamma,\bar{x}_{\gamma+\delta\gamma})-Z_{\nu}(\gamma,\bar{x}_{\gamma})\right)
\]
where $\bar{x}_{\gamma}$ and $\bar{x}_{\gamma+\delta\gamma}$ now denote the parameter $\bar{x}$ for distributions with parameter $\gamma$ and ${\gamma+\delta\gamma}$.  
Let us begin with the continuous case. If $\delta\mathcal{D}$ denotes the domain increment associated to the variation $\delta\gamma$, it remains
\begin{align}
\frac{dZ_{\nu}(\gamma,\bar{x}_{\gamma})}{d\gamma} & =\int_{\mathcal{D}}\frac{d}{d\gamma}\left(1+\gamma\left(x-\bar{x}_{\gamma}\right)\right)^{\nu}Q(x)\dx \label{eq:dZwithoutincrement}\\
 & +\lim_{\delta\gamma\rightarrow0}\frac{1}{\delta\gamma}\int_{\mathcal{\delta\mathcal{D}}}\left(1+(\gamma+\delta\gamma)\left(x-\bar{x}_{\gamma+\delta\gamma}\right)\right)^{\nu}Q(x) \dx \label{eq:dZwithincrement}
\end{align}
%\end{align}
 Of course, when $\mathcal{D}$ does not depend on $\gamma$, we only
have the first term, and it is easy to obtain (\ref{eq:Zderivate}). Otherwise, in order to satisfy the positivity of the integrand, the domain $\mathcal{D}$ is bounded above by $\left(\bar{x}_{\gamma}-\frac{1}{{\gamma}}\right)$ for $\gamma<0$ and below by the same value for $\gamma>0$. Then, the second integral, say $G$, can be expressed as
\def\sign#1{{\mathrm{sign}\left(#1\right)}}
\begin{align}
G & =\sign{\gamma}\int_{\bar{x}_{\gamma+\delta\gamma}-\frac{1}{{\gamma+\delta\gamma}}}^{\bar{x}_{\gamma}-\frac{1}{{\gamma}}}\left(1+(\gamma+\delta\gamma)\left(x-\bar{x}_{\gamma+\delta\gamma}\right)\right)^{\nu}Q(x)\dx\\
 & =\frac{\sign{\gamma}}{\gamma+\delta\gamma}\int_{0}^{a}y^{\nu}Q\left(\frac{y-1}{{\gamma+\delta\gamma}}+\bar{x}_{\gamma+\delta\gamma}\right)\mathrm{d}y
\end{align}
with $a=(\gamma+\delta\gamma)\left(\bar{x}_{\gamma+\delta\gamma}-\bar{x}_{\gamma}\right)-\frac{\delta\gamma}{\gamma}$,
that tends to zero with $\delta\gamma$ if $\bar{x}_{\gamma}$ is
continuous.
% Therefore, the value of the integral depends on the behaviour
% of the integrand for $y\rightarrow0$.
At first order,
% with the hypothesis that $Q\left(\bar{x}_{\gamma+\delta\gamma}-\frac{1}{{\gamma+\delta\gamma}}\right)\neq0$,
we then obtain
\[
G=\sign{\gamma}\frac{Q\left(\bar{x}_{\gamma+\delta\gamma}-\frac{1}{{\gamma+\delta\gamma}}\right)}{\gamma+\delta\gamma}
\int_{0}^{a}y^{\nu}\mathrm{d}y\propto
\frac{a^{1+\nu}}{1+\nu}
\]
for $\nu>-1$.
%, and the integral diverges otherwise, that would of
%course mean that $Z_{\nu}(\gamma+\delta\gamma,\bar{x}_{\gamma+\delta\gamma})$
%is not defined.
Then, it is readily checked that $\lim_{\delta\gamma\rightarrow0}\frac{1}{\delta\gamma}G=0$ for $\nu>0$, so that (\ref{eq:dZwithincrement}) is always zero for $\nu>0$ and (\ref{eq:Zderivate}) is true.
%
% \begin{align}
%  & \lim_{\delta\gamma\rightarrow0}\frac{1}{\delta\gamma}G\propto\lim_{\delta\gamma\rightarrow0}\frac{1}{\delta\gamma}{a^{1+\nu}}=\lim_{\delta\gamma\rightarrow0}{\delta\gamma}^{\nu}\left(\frac{a}{\delta\gamma}\right)^{1+\nu}\\
%  & =\lim_{\delta\gamma\rightarrow0}{\delta\gamma}^{\nu}\left(\left(\gamma+\delta\gamma\right)\frac{\left(\bar{x}_{\gamma+\delta\gamma}-\bar{x}_{\gamma}\right)}{\delta\gamma}-\frac{1}{\gamma}\right)^{1+\nu}\\
%  & =0\text{~~~ for $\nu>0$.}\end{align}
%  Therefore, we always have \[
% \frac{dZ_{\nu}(\gamma,\bar{x})}{d\gamma}=\int_{\mathcal{D}}\frac{d}{d\gamma}\left(1+\gamma\left(x-\bar{x}_{\gamma}\right)\right)^{\nu}Q(x)\dx\]
%
~\\
In the discrete case, the partition function is
\[
Z_{\nu}(\gamma,\bar{x}_{\gamma})=\sum_{x\in\mathcal{D}}\left(1+\gamma\left(x-\bar{x}_{\gamma}\right)\right)^{\nu}Q(x)
\]
%  where $\mathcal{D}_{\gamma}=\{x\in\mathbb{N}:\left(1+\gamma\left(x-\bar{x}_{\gamma}\right)\right)\ge0\}$.

There exists singular isolated values of $\gamma$ such that $1+\gamma\left(x-\bar{x}_{\gamma}\right)=0$, for $x$ integer. For such values, the corresponding term in the partition function diverges for $\nu<0$. Contrary to the continuous case where the domain of $\gamma$ is contiguous,  the domain of values of $\gamma$ ensuring that the partition function is finite will be interrupted by isolated values of $\gamma$: the domain of possible $\gamma$ will be constituted of segments.
%
% Préciser que pour $\nu<0$, il peut y avoir divergence, singularités,
% pour certains gamma et que donc on a des subset de gamma ou cela converge.
% La différence avec le cas continu est que la divergence intervient
% pour des gamma isolés, ce qui permet la convergence pour des sous-ensembles
% (segments) de valeurs de $\gamma$ alors que c'est pour un domaine
% continu

As in the continous case, the derivative of the partition function
writes as the sum of two terms, the second one involving a domain
increment \begin{align}
\frac{dZ_{\nu}(\gamma,\bar{x}_{\gamma})}{d\gamma} & =\sum_{\mathcal{D}}\frac{d}{d\gamma}\left(1+\gamma\left(x-\bar{x}_{\gamma}\right)\right)^{\nu}Q(x) \label{eq:Zderivate_discret}\\
 & +\lim_{\delta\gamma\rightarrow0}\frac{1}{\delta\gamma}\sum_{\mathcal{\delta\mathcal{D}}}\left(1+(\gamma+\delta\gamma)\left(x-\bar{x}_{\gamma+\delta\gamma}\right)\right)^{\nu}Q(x) \label{eq:Zderivete_increment_discret}
\end{align}
If $\mathcal{D}$ does not depend on $\gamma$, there is no domain increment and
the derivative is given by (\ref{eq:Zderivate_discret}). When the bounds of $ \mathcal{D}$ depend of $\gamma$, the domain increment is given by the integers in the interval
$\left(\lceil\bar{x}_{\gamma+\delta\gamma}-\frac{1}{{\gamma+\delta\gamma}}\rceil,\lceil{\bar{x}_{\gamma}-\frac{1}{{\gamma}}}\rceil\right)$
( $\gamma>0$) or $\left(\lfloor{\bar{x}_{\gamma}-\frac{1}{{\gamma}}}\rfloor,\lfloor\bar{x}_{\gamma+\delta\gamma}-\frac{1}{{\gamma+\delta\gamma}}\rfloor\right)$
($\gamma<0$); where $\lfloor x\rfloor$ is the floor function that
returns the largest integer less than or equal to x; and $\lceil x\rceil$
is the ceil function, the smallest integer not less than $x$. If $\gamma$ belongs in some interval such that the domain increment remains empty, then
the derivative is of course simply (\ref{eq:Zderivate_discret}). An extension will occur for an infinitesimal variation $\delta\gamma$ if  $\bar{x}_{\gamma}-\frac{1}{{\gamma}}$
is precisely an integer, say $k$,

% Otherwise,
% if $k$ denotes an integer in the extension, with
% $\left(\lceil\bar{x}_{\gamma+\delta\gamma}-\frac{1}{{\gamma+\delta\gamma}}\rceil,\lceil{\bar{x}_{\gamma}-\frac{1}{{\gamma}}}\rceil\right)$
% $\lceil{\bar{x}_{\gamma}-\frac{1}{{\gamma}}}\rceil > k > \lceil\bar{x}_{\gamma+\delta\gamma}-\frac{1}{{\gamma+\delta\gamma}}\rceil
%  \left(\lceil\bar{x}_{\gamma+\delta\gamma}-\frac{1}{{\gamma+\delta\gamma}}\rceil,\lceil{\bar{x}_{\gamma}-\frac{1}{{\gamma}}}\rceil\right)$
% for ( $\gamma>0$)
%
%
%  $k=\bar{x}_{\gamma}-\frac{1}{{\gamma}}+\epsilon$,
%
%

% Et si gamma reste dans un subset, qui est open, alors $\delta\mathcal{D}=\emptyset$,
% et la dérivée est facile. Maintenant, si gamma est sur la frontiÃ¨re,
% ce qui suppose bien sur $\nu>0$, on aura une extension. et lÃ , la
% dérivée n'est finie que pour $nu>1$

% Here, the domain increment is given by the integers in the interval
% $\left[\lceil\bar{x}_{\gamma+\delta\gamma}-\frac{1}{{\gamma+\delta\gamma}}\rceil,\lceil{\bar{x}_{\gamma}-\frac{1}{{\gamma}}}\rceil\right]$
% ( $\gamma>0$) or $\left[\lfloor{\bar{x}_{\gamma}-\frac{1}{{\gamma}}}\rfloor,\lfloor\bar{x}_{\gamma+\delta\gamma}-\frac{1}{{\gamma+\delta\gamma}}\rfloor\right]$
% ($\gamma<0$); where $\lfloor x\rfloor$ is the floor function that
% returns the largest integer less than or equal to x; and $\lceil x\rceil$
% is the ceil function, the smallest integer not less than $x$. We
% will have a domain extension if $\bar{x}_{\gamma}-\frac{1}{{\gamma}}$
% is precisely an integer, say $k$, so that the small variation $\delta\gamma$
% possibly involve such extension.
%
Then,  the second sum reduces to
% \[
% G=\left(1+(\gamma+\delta\gamma)\left(k-\bar{x}_{\gamma+\delta\gamma}\right)\right)^{\nu}Q(k)\]
%  and since $k=\bar{x}_{\gamma}-\frac{1}{{\gamma}}$, we still have
\begin{align}
G= & \left(1+(\gamma+\delta\gamma)\left(k-\bar{x}_{\gamma+\delta\gamma}\right)\right)^{\nu}Q(k)\\
= & \left(-\frac{\delta\gamma}{\gamma}-(\gamma+\delta\gamma)\left(\bar{x}_{\gamma+\delta\gamma}-\bar{x}_{\gamma} \right)\right)^{\nu}Q(k),\end{align}
and finally \begin{align}
 & \lim_{\delta\gamma\rightarrow0}\frac{1}{\delta\gamma}G
%\propto\lim_{\delta\gamma\rightarrow0}\frac{1}{\delta\gamma}\left(-\frac{\delta\gamma}{\gamma}-(\gamma+\delta\gamma)\left(\bar{x}_{\gamma+\delta\gamma}-\bar{x}_{\gamma}\right)\right)^{\nu}Q(k)\\
% &
 =\lim_{\delta\gamma\rightarrow0}{\delta\gamma}^{\nu-1}\left(\left(\gamma+\delta\gamma\right)\frac{\left(\bar{x}_{\gamma+\delta\gamma}-\bar{x}_{\gamma}\right)}{\delta\gamma}-\frac{1}{\gamma}\right)^{1+\nu}
 =0\text{~~~ for $\nu>1$.}\end{align}
since all terms in the parenthesis remains finite when $\delta\gamma\rightarrow 0$. In such case the derivative reduces to (\ref{eq:Zderivate_discret}) and (\ref{eq:Zderivate}) is true.

\bibliographystyle{elsart-num-sort}%{ieeetr}
\bibliography{sub_bib}

\begin{thebibliography}{10}
\expandafter\ifx\csname url\endcsname\relax
  \def\url#1{\texttt{#1}}\fi
\expandafter\ifx\csname urlprefix\endcsname\relax\def\urlprefix{URL }\fi

\bibitem{Asadi2006}
M.~Asadi, I.~Bayramoglu, The mean residual life function of a k-out-of-n
  structure at the system level, IEEE Transactions on Reliability 55 (2006)
  314--318.

\bibitem{Banerjee2005}
A.~Banerjee, S.~Merugu, I.~S. Dhillon, J.~Ghosh, Clustering with {B}regman
  divergences, J. Mach. Learn. Res 6 (2005) 1705--1749.

\bibitem{Baraniuk2001}
R.~Baraniuk, P.~Flandrin, A.~Janssen, O.~Michel, Measuring time-frequency
  information content using the {R}{\'e}nyi entropies, IEEE Transactions on
  Information Theory 47 (2001) 1391--1409.

\bibitem{Bashkirov2004}
A.~G. Bashkirov, On maximum entropy principle, superstatistics, power-law
  distribution and {R}{\'e}nyi parameter, Physica A 340 (2004) 153--162.

\bibitem{Basseville1989}
M.~Basseville, Distance measures for signal processing and pattern recognition,
  Signal Processing 18 (1989) 349--369.

\bibitem{Beck2004}
C.~Beck, Generalized statistical mechanics of cosmic rays, Physica A 331 (2004)
  173--181.

\bibitem{Bhandari1993}
D.~Bhandari, N.~R. Pal, Some new information measures for fuzzy sets,
  Information Sciences 67 (1993) 209--228.

\bibitem{Cebrian2003}
A.~C. Cebrian, M.~Denuit, P.~Lambert, Generalized pareto fit to the society of
  actuaries' large claims database, North American Actuarial Journal 7 (2003)
  18--36.

\bibitem{Csiszar1991}
I.~Csisz{\'a}r, Why least squares and maximum entropy? an axiomatic approach to
  inference for linear inverse problems, Annals of Statistics 19 (1991)
  2032--2066.

\bibitem{Csiszar1995}
I.~Csisz{\'a}r, Generalized cutoff rates and {R}{\'e}nyi's information
  measures, IEEE Transactions on Information Theory 41 (1995) 26--34.

\bibitem{Ellis1985}
R.~S. Ellis, Entropy, Large Deviations, and Statistical Mechanics, vol. 271 of
  Grundlehren der mathematischen Wissenschaften, Springer-Verlag, 1985.

\bibitem{Esteban1995}
M.~D. Esteban, Divergence statistics based on entropy functions and stratified
  sampling, Information Sciences 87 (1995) 185--203.

\bibitem{Golan2002}
A.~Golan, J.~M. Perloff, Comparison of maximum entropy and higher-order entropy
  estimators, Journal of Econometrics 107 (2002) 195 -- 211.

\bibitem{Grendar2004a}
M.~Grendar, M.~Grendar, Maximum entropy method with non-linear moment
  constraints: challenges, AIP, 2004.

\bibitem{He2003}
Y.~He, A.~Hamza, H.~Krim, A generalized divergence measure for robust image
  registration, IEEE Transactions on Signal Processing, [see also Acoustics,
  Speech, and Signal Processing 51 (2003) 1211--1220.

\bibitem{Jaynes1957}
E.~T. Jaynes, Information theory and statistical mechanics, Phys. Rev. 108
  (1957) 171.

\bibitem{Jaynes1982}
E.~T. Jaynes, On the rationale of maximum entropy methods, Proc. IEEE 70 (1982)
  939--952.

\bibitem{Jizba2004}
P.~Jizba, T.~Arimitsu, The world according to {R}{\'e}nyi: thermodynamics of
  multifractal systems, Annals of Physics 312 (2004) 17--59.

\bibitem{Krishnamachari2004}
A.~Krishnamachari, V.~moy Mandal, Karmeshu, Study of dna binding sites using
  the r{\'e}nyi parametric entropy measure, Journal of Theoretical Biology 227
  (2004) 429--436.

\bibitem{Kullback1959}
S.~Kullback, Information Theory and Statistics, Wiley, New York, 1959.

\bibitem{LaCour2004}
B.~LaCour, Statistical characterization of active sonar reverberation using
  extreme value theory, Oceanic Engineering, IEEE Journal of 29 (2004)
  310--316.

\bibitem{Mayoral1998}
M.~M. Mayoral, {R}{\'e}nyi's entropy as an index of diversity in simple-stage
  cluster sampling, Information Sciences 105 (1998) 101--114.

\bibitem{Molina2007}
I.~Molina, D.~Morales, {R}{\'e}nyi statistics for testing hypotheses in mixed
  linear regression models, Journal of Statistical Planning and Inference 137
  (2007) 87--102.

\bibitem{Montfort1986}
M.~A. J.~V. Montfort, J.~V. Witter, Generalized {P}areto distribution applied
  to rainfall depths, Hydrological Sciences Journal 31 (1986) 151--162.

\bibitem{Nadarajah2003}
S.~Nadarajah, K.~Zografos, Formulas for {R}{\'e}nyi information and related
  measures for univariate distributions, Information Sciences 155 (2003)
  119--138.

\bibitem{Nadarajah2005}
S.~Nadarajah, K.~Zografos, Expressions for {R}{\'e}nyi and shannon entropies
  for bivariate distributions, Information Sciences 170 (2005) 173--189.

\bibitem{Nanda2007}
A.~K. Nanda, S.~S. Maiti, {R}{\'e}nyi information measure for a used item,
  Information Sciences 177 (2007) 4161--4175.

\bibitem{Naudts2002}
J.~Naudts, Dual description of nonextensive ensembles, Chaos, Solitons, and
  Fractals 13 (2002) 445--450.

\bibitem{Neemuchwala2005}
H.~Neemuchwala, A.~Hero, P.~carson, Image matching using alpha-entropy measures
  and entropic graphs, Signal Processing 85 (2005) 277--296.

\bibitem{Nock2006}
R.~Nock, F.~Nielsen, On weighting clustering, IEEE Trans. Pattern Anal. Mach.
  Intell 28 (2006) 1223--1235.

\bibitem{Raggio1999a}
G.~A. Raggio, On equivalence of thermostatistical formalisms,
  http://arxiv.org/abs/cond-mat/9909161 (1999).

\bibitem{Renyi1961}
A.~R{\'e}nyi, On measures of entropy and information, Univ. California Press,
  Berkeley, Calif., 1961.

\bibitem{Song2001}
K.-S. Song, {R}{\'e}nyi information, loglikelihood and an intrinsic
  distribution measure, Journal of Statistical Planning and Inference 93 (2001)
  51--69.

\bibitem{Tsallis1988}
C.~{T}sallis, Possible generalization of boltzmann-gibbs statistics, Journal of
  Statistical Physics 52 (1988) 479--487.

\bibitem{Tsallis2002}
C.~{T}sallis, Entropic nonextensivity: a possible measure of complexity, Chaos,
  Solitons,\& Fractals 13 (2002) 371--391.

\bibitem{Tsallis1998}
C.~{T}sallis, R.~S. Mendes, A.~R. Plastino, The role of constraints within
  generalized nonextensive statistics, Physica A 261 (1998) 534--554.

\bibitem{Vignat2004}
C.~Vignat, A.~Hero, J.~A. Costa, About closedness by convolution of the
  {T}sallis maximizers, Physica A 340 (2004) 147--152.

\bibitem{Vinga2004}
S.~Vinga, J.~S. Almeida, R{\'e}nyi continuous entropy of {DNA} sequences,
  Journal of Theoretical Biology 231 (2004) 377--388.

\end{thebibliography}

\end{document}